\documentclass[aps,prl,twocolumn,superscriptaddress,longbibliography,reprint,floatfix]{revtex4-2}

\usepackage{amsmath}
\usepackage{amssymb}
\usepackage{amsfonts}
\usepackage{amsthm}
\usepackage{mathtools}
\usepackage{graphicx}
\usepackage[dvipsnames]{xcolor}
\usepackage{physics}
\usepackage{bm}
\usepackage{bbm}
\usepackage{dsfont}
\usepackage[normalem]{ulem}
\usepackage{enumitem}
\usepackage{multirow}
\usepackage{thmtools,thm-restate}
\usepackage[export]{adjustbox}

\usepackage{algpseudocode}

\newcounter{protocol}
\renewcommand{\theprotocol}{\arabic{protocol}}
{%
  \refstepcounter{protocol}%
  \par\medskip\noindent
  \textbf{Protocol~\theprotocol%
  \ifx\relax#1\relax\else\ (\textit{#1})\fi.}%
  \par\smallskip\noindent\rule{\columnwidth}{0.4pt}\par\smallskip
}%
{%
  \smallskip\noindent\rule{\columnwidth}{0.4pt}\par\medskip
}

\newtheorem{theorem}{Theorem}

\newcommand{\sgn}{\operatorname{sgn}}
\newcommand{\id}{\mathbbm 1}

\usepackage{times}
\usepackage{booktabs,placeins,needspace,microtype}
\definecolor{myrefcolor}{rgb}{0.067,0.5,0.5}
\usepackage[
    breaklinks,
    pdftex,
    colorlinks=true,
    linkcolor=myrefcolor,
    citecolor=myrefcolor,
    urlcolor=myrefcolor
]{hyperref}

\hypersetup{pdftitle={Sudden death of entanglement, rebirth of magic},pdfauthor={Chenfeng Cao}}

\begin{document}

\title{Sudden death of entanglement, rebirth of magic}
\author{Chenfeng Cao}
\email{chenfeng.cao@fu-berlin.de}
\affiliation{Dahlem Center for Complex Quantum Systems, Freie Universit\"at Berlin, 14195 Berlin, Germany}
\affiliation{HK Institute of Quantum Science \& Technology, The University of Hong Kong, Hong Kong, China}

\begin{abstract}
Local Markovian noise cannot bring entanglement back, but it can bring magic back.
Local dissipation can push a state out of the stabilizer polytope as well as in.
Under independent amplitude damping, the $n$-qubit state $\alpha\ket{0^n}+\beta\ket{1^n}$ with $0<\alpha<\beta$ loses its magic for a finite interval.
The coherence between its ground and fully excited configurations keeps falling throughout. Magic returns when the fully excited population falls below that coherence.
The damping probabilities for entanglement death and magic rebirth obey $\gamma_e+\gamma_+=1$ for every $n\ge2$.
Magic returns to the register just as its reservoirs become entangled.
For a sufficiently small ground-state amplitude the register is already fully separable at rebirth.
Every proper marginal is a stabilizer mixture throughout, yet a joint parity measurement followed by Clifford decoding can concentrate the reborn magic onto one qubit.
With thermal relaxation and dephasing, we find a range in which magic loss is permanent for every isolated qubit, yet two qubits exposed to the same noise can regain it.
\end{abstract}
\maketitle

\section{Introduction}
Irreversibility under local Markovian noise is resource-dependent. Once a pair of qubits is separable, every subsequent local channel keeps it separable~\cite{Zyczkowski2001Dynamics,Yu2009Sudden,Aolita2015Open}. The classic example is the sudden death of entanglement in the state $\alpha\ket{00}+\beta\ket{11}$ under independent amplitude damping~\cite{Yu2004Finite}, observed in experiments with pairs of photons~\cite{Almeida2007Environment,Salles2008Experimental}. Entanglement vanishes while both qubits still have a substantial probability of occupying their excited state. Entanglement can return if the reservoirs have memory~\cite{Bellomo2007NonMarkovian}. Here we ask what can return even when they have none.

Magic is the resource that allows suitable input states to promote stabilizer circuits to universal quantum computation~\cite{Gottesman1999The,Aaronson2004Improved,Bravyi2005Universal,Veitch2014The,Howard2017Application}. A state has magic when it lies outside the convex hull of pure stabilizer states. Local relaxation can drive a state out of this polytope as well as in~\cite{Trigueros2025Nonstabilizerness,Junior2026Trading}. For example, a qubit prepared in $\ket{+}$ starts on the stabilizer octahedron. At damping probability $\gamma$ its Bloch vector is $(\sqrt{1-\gamma},0,\gamma)$, outside the octahedron for every $0<\gamma<1$. The stronger question is whether magic can disappear for a finite interval and return along the same trajectory, even after the whole register has become separable.

Figure~\ref{fig:mechanism}a shows why the shape of the trajectory matters. A straight path ending at a stabilizer state can enter the polytope only once. Under amplitude damping the real GHZ state $\alpha\ket{0^n}+\beta\ket{1^n}$, with $0<\alpha<\beta$, follows a curved path. As the ground state fills, the register first loses its magic. Later, the fully excited population falls below its coherence with the ground state, and magic returns. The damping probabilities for entanglement death and magic rebirth are reflected about $\gamma=1/2$ (Fig.~\ref{fig:mechanism}b). Interchanging a qubit with its reservoir exchanges excitation survival and loss, so the return of magic marks the onset of entanglement in the reservoirs. For small enough $\alpha$ the register is then fully separable. Tracing out even one qubit erases the coherence, yet a joint parity measurement can recover the magic on a single decoded qubit.

Pure amplitude damping already lets a single qubit lose and regain magic. Thermal excitation and dephasing let us ask when two qubits are essential. A channel can preserve all single-qubit stabilizer states and still generate magic on one half of a pair if the other half is protected~\cite{Seddon2019Quantifying,CaoZhao2026Records}. Here both qubits relax throughout and the reservoirs go unobserved. We find the exact range in which either qubit alone would stay inside the stabilizer octahedron after entering it, while the pair can leave again. We also give a late-time test based on the initial coherences of a general register.

\begin{figure}[tbp]
\centering
\includegraphics[width=\columnwidth]{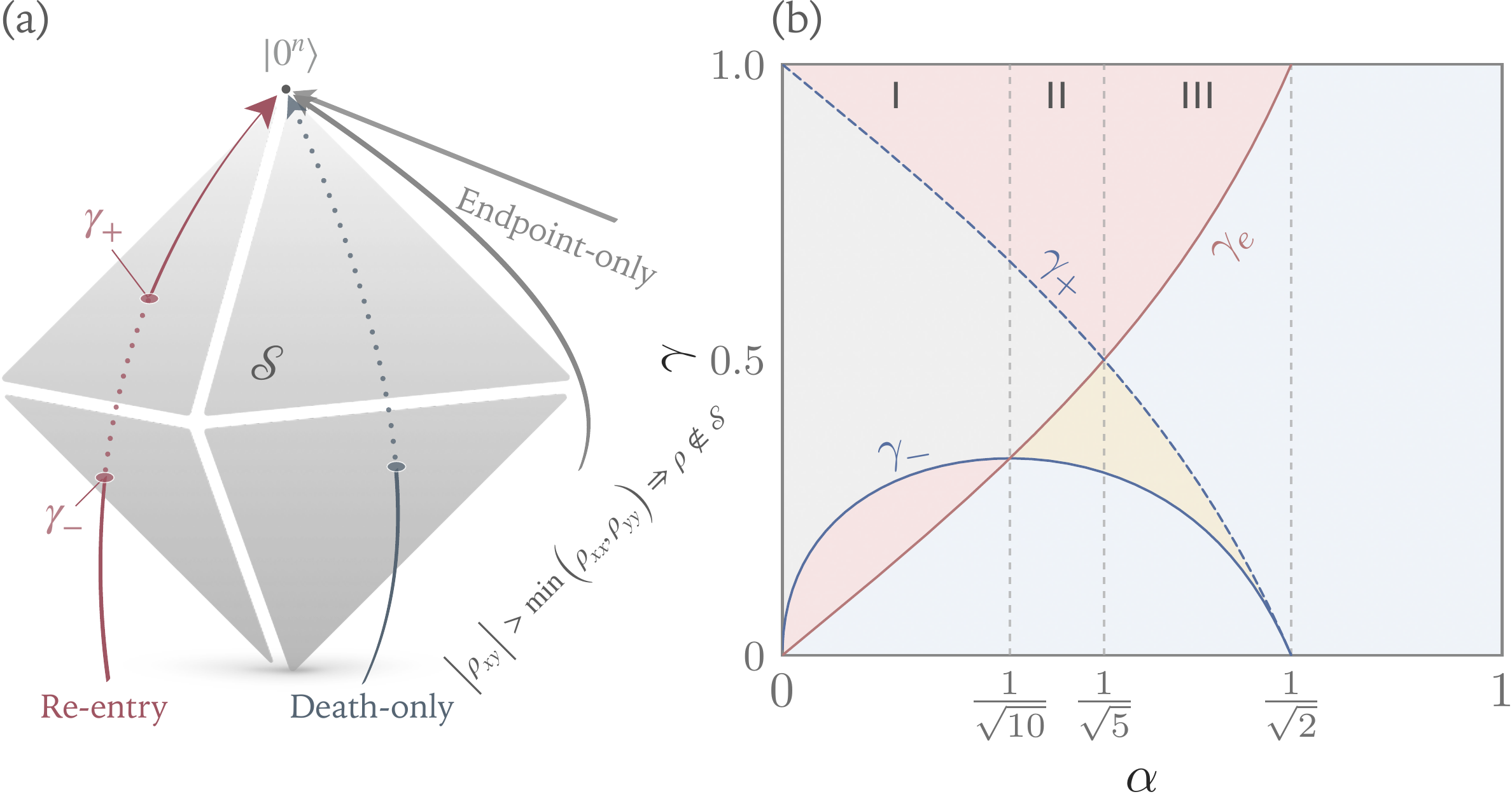}
\caption{\textbf{Magic and entanglement under local amplitude damping.}
\textbf{a} The red trajectory enters and leaves the stabilizer polytope $\mathcal S$. The blue trajectory stays inside after entering. The gray trajectories reach the boundary only at the ground state.
\textbf{b} For the two-qubit cat $\alpha\ket{00}+\beta\ket{11}$, the blue curves mark magic death ($\gamma_-$) and rebirth ($\gamma_+$). The red curve marks entanglement death ($\gamma_e$), which falls before, within or after the interval without magic in regions I, II and III. Blue, pink, yellow and gray indicate both resources, magic alone, entanglement alone and neither, respectively.}
\label{fig:mechanism}
\end{figure}

\section{Two population crossings bound an interval without magic}
\label{sec:setup}

Consider a register prepared in the real cat state
$\ket{\psi_n}=\alpha\ket{0^n}+\beta\ket{1^n}$, with $n\ge2$,
$0<\alpha<\beta$ and $\alpha^2+\beta^2=1$. We write $r=\alpha/\beta$
for the amplitude ratio. Under pure amplitude damping, each qubit independently
loses its excitation with probability $\gamma=1-e^{-t/T_1}$, where $T_1$
is the population relaxation time. Write $q=1-\gamma$ for the survival
probability. The register has only one pair of off-diagonal entries, connecting
the ground and fully excited configurations. The populations of these
configurations and the coherence between them are
\begin{equation}
 \begin{aligned}
 P_0&=\alpha^2+\beta^2\gamma^n,\qquad P_n=\beta^2q^n,\\
 c&=\langle0^n|\rho_n|1^n\rangle=\alpha\beta q^{n/2}.
 \end{aligned}
 \label{eq:rho_n}
\end{equation}
An initially excited register contributes to $P_0$ when all $n$ excitations
are lost and to $P_n$ when all $n$ survive. The coherence decays with the
amplitude for this second event.

These three entries suffice to test for magic. If $c\le\min(P_0,P_n)$, the state is a mixture of
$(\ket{0^n}+\ket{1^n})/\sqrt2$ with weight $2c$ and computational-basis
states carrying the remaining weight.
Conversely, a stabilizer mixture cannot sustain a coherence whose modulus
exceeds the smaller of the two populations it connects (Methods).
For a quantitative measure, we take the robustness of magic $\mathcal R$,
the minimum sum of absolute weights in a signed decomposition into pure
stabilizer states~\cite{Howard2017Application}. For a Pauli expectation after a Clifford circuit, sampling an optimal decomposition gives a sample bound proportional to $\mathcal R^2$ at fixed additive precision and confidence~\cite{Pashayan2015Estimating,Howard2017Application,Seddon2021Quantifying}. The excess $\mathcal R-1$ vanishes precisely for stabilizer mixtures. On the cat trajectory it is
\begin{equation}
 \mathcal R(\rho_n)-1=2[c-\min(P_0,P_n)]_+,\qquad [y]_+=\max(0,y).
 \label{eq:RoM_GHZ}
\end{equation}
Death and rebirth are therefore the same for every measure that vanishes exactly on stabilizer mixtures. For stabilizer R\'enyi entropies, defined through Pauli moments on pure states~\cite{Leone2022Stabilizer}, a faithful mixed-state extension requires a separate construction, such as minimization over pure-state ensembles~\cite{Leone2024Stabilizer}. The cat formula appeared in an earlier version of this work~\cite{Cao2026RebirthV2}, and Supplementary Note~\ref{si:membership} gives a direct derivation. De Oliveira Junior \emph{et al.} derived the general $X$-state formula~\cite{deOliveira2026Symmetry}. Its decomposition into parity sectors is a special case of the syndrome identity for stabilizer codes~\cite{Bozzetto2026Coherent}.

Initially the coherence exceeds $P_0$. As the ground configuration fills and
$c$ falls, the two meet at a unique damping probability $\gamma_-$, at which
the register loses its magic (red trajectory in Fig.~\ref{fig:mechanism}a). Both the coherence and the fully excited population keep falling, but the population decays twice as fast, with $P_n\propto e^{-nt/T_1}$ and $c\propto e^{-nt/(2T_1)}$. Magic returns when they cross at
\begin{equation}
 \gamma_+=1-r^{2/n}.
 \label{eq:gp_n}
\end{equation}
There $P_n=c=\alpha^2<P_0$, so the two crossings are distinct. The register is a stabilizer mixture throughout $[\gamma_-,\gamma_+]$. On the ensuing interval $\gamma_+<\gamma<1$, which we call the \emph{reborn branch}, the coherence exceeds $P_n$ and the register again has magic. A balanced cat ($r=1$) starts as a stabilizer state and acquires magic as soon as damping begins. For $r>1$ the coherence exceeds $P_n$ at every finite time, so magic persists throughout.

The duration of magic death and the height of the later peak scale differently with register size. At fixed $r$, $\gamma_+-\gamma_-$ decreases faster
than exponentially with $n$ (Fig.~\ref{fig:deadwindow}a). Once the state
has left the stabilizer polytope again, its excess is particularly simple
as a function of the total exposure $x=-n\ln q=nt/T_1$,
\begin{equation}
 \mathcal R(\rho_n)-1=2\beta e^{-x/2}(\alpha-\beta e^{-x/2}).
 \label{eq:exposure_collapse}
\end{equation}
At every register size it reaches $\alpha^2/2$ at $x_\star=2\ln(2/r)$
(Fig.~\ref{fig:deadwindow}b). The peak time in the laboratory is therefore
$T_1x_\star/n$. The excess stays above half of its peak over an exposure interval of length $4\ln(1+\sqrt2)\simeq3.53$, that is, for a laboratory time of $3.53\,T_1/n$. For $\alpha=0.4$ the peak excess is $0.080$ and, with ten
qubits, occurs at $\gamma=0.262$. An initial excitation then survives on each site with probability $0.738$, and the ground configuration has a population of about $0.16$. One can also compare the peak with the magic that
went into preparing the cat. Dividing it by the initial excess
$2\alpha(\beta-\alpha)$ gives $r/[4(1-r)]$, about one fifth for this example.
For $4/5<r<1$, relaxation supplies enough magic that the reborn peak exceeds
the initial excess.

\begin{figure}[tbp]
\centering
\includegraphics[width=.97\columnwidth]{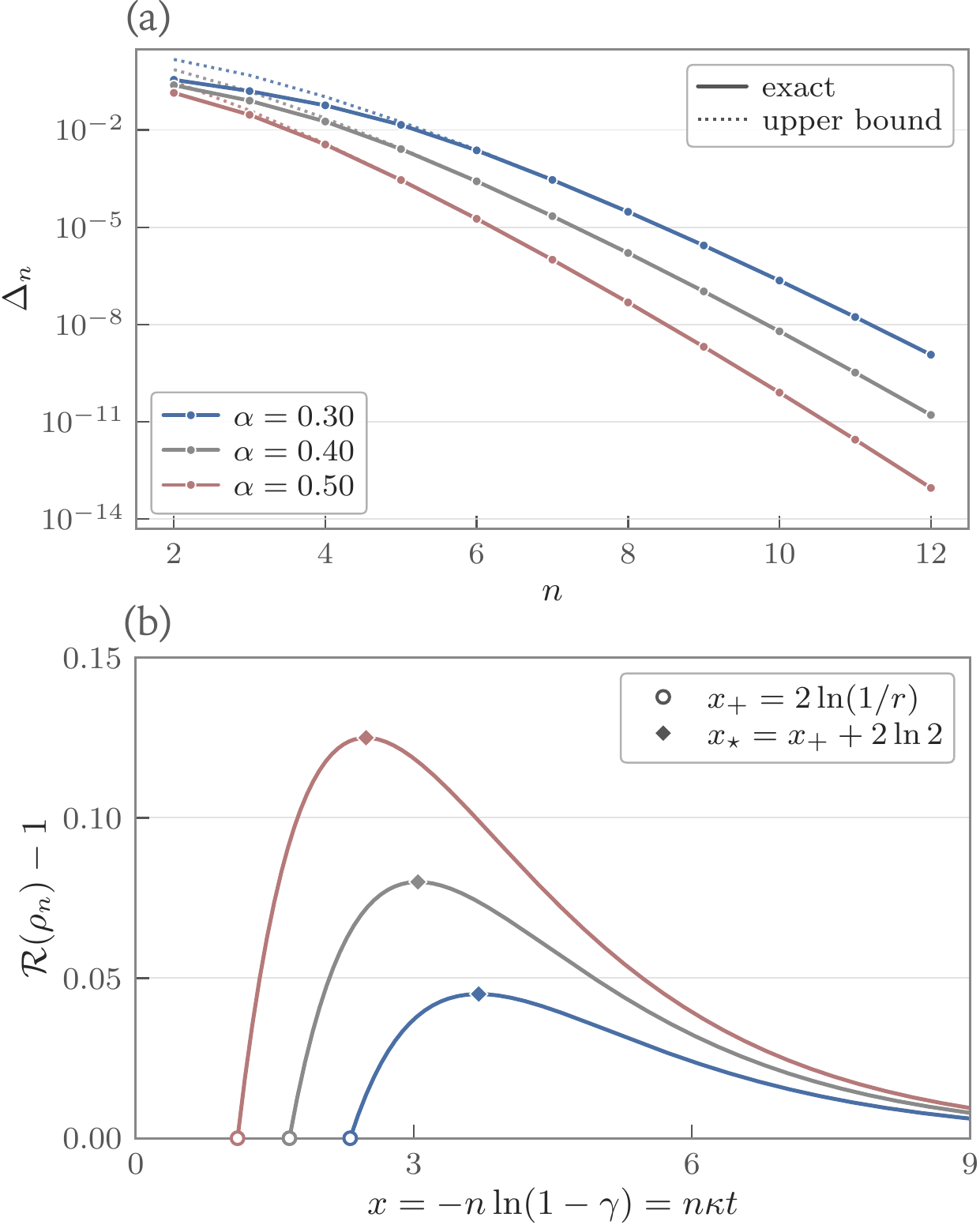}
\caption{\textbf{Duration of magic death and size of the reborn signal.}
\textbf{a} The interval $\Delta_n=\gamma_+-\gamma_-$ for different ground
amplitudes $\alpha$ of a pure cat, with $r=\alpha/\sqrt{1-\alpha^2}$.
Dotted curves give the upper bound in Methods. \textbf{b} After rebirth the
robustness excess for a given $\alpha$ follows Eq.~\eqref{eq:exposure_collapse}
at every register size, with the same colors as in \textbf{a}. The horizontal coordinate is $x=n\kappa t=nt/T_1$,
where $\kappa=1/T_1$. Circles mark the return at $x_+=2\ln(1/r)$ and
diamonds the peak at $x_\star=x_++2\ln2$.}
\label{fig:deadwindow}
\end{figure}

At each fixed register size, an open neighborhood around the cat input, including
states with additional coherences, also shows death followed by rebirth (Supplementary Note~\ref{si:membership}). A controlled change
in the preparation illustrates how the duration can be tuned. Applying $Z$
to a single qubit with probability $(1-\mu)/2$ reduces the initial coherence
to $\mu\alpha\beta$ while preserving the populations. For $r<\mu<1$ the
state still starts with magic, and at every $n$ it is a stabilizer mixture
throughout
\begin{equation}
 2\ln(\mu/r)\le x\le 2\ln[1/(\mu r)].
 \label{eq:visibility_window}
\end{equation}
The upper endpoint is the exact return, where $c=P_n$. For $\alpha=0.4$, a
$5\%$ phase-flip probability on one qubit gives an interval of at least
$0.421\,T_1/n$ in laboratory time and a peak excess of $0.0648$. The same visibility of the whole register is
assumed when comparing sizes. Phase errors on every qubit can reduce it
further as the register grows.

Pure amplitude damping already permits a single qubit to lose its magic and
regain it. What is distinctive about the cat is that tracing out even one
qubit removes $c$, so every proper marginal is diagonal throughout. To find
when more than one qubit is necessary for rebirth, we now allow the
reservoir to have a thermal population and the qubits to undergo dephasing.

\section{The minimum register size for magic rebirth}
\label{sec:thermal_rebirth}\label{sec:channelcond}

Each qubit now relaxes independently towards the same equilibrium
excited-state population $\vartheta\in[0,1/2]$. Its population difference
from equilibrium decays as $q=e^{-t/T_1}$ and its coherence as
$e^{-t/T_2}=q^a$, where $a=T_1/T_2$. Pure dephasing at rate $\Gamma_\phi$ gives $1/T_2=1/(2T_1)+\Gamma_\phi$, so $a\ge1/2$. Pure amplitude damping is the point $(a,\vartheta)=(1/2,0)$. We take the rates to be constant and define the Pauli axes in the rotating frame of the controls, with no residual detuning. A rotation about $Z$ alone could carry a state across a face of the stabilizer octahedron. Here the coherence factor $q^a$ is real and positive. For the cat, $c(t)=q^{an}c(0)$ keeps its initial phase and decreases in magnitude throughout the trajectory.

\Needspace{9\baselineskip}
\begin{theorem}[Minimum register size]
\label{thm:min_rebirth_size}\label{thm:channels}
Let all qubits undergo this relaxation for the same duration in the specified
rotating frame, with the reservoirs unobserved. Let $n_{\min}$ be the smallest
register size for which some initial state has magic, lies in the stabilizer
polytope for an open time interval and again has magic at a later finite time.
Then
\begin{equation}
\label{eq:min_rebirth_size}
 n_{\min}=\begin{cases}
  1,&a<1-2\vartheta,\\
  2,&1-2\vartheta\le a<1-\vartheta,\\
  \infty,&a\ge1-\vartheta.
 \end{cases}
\end{equation}
Here $\infty$ means that rebirth is impossible for every finite register.
\end{theorem}

The boundaries in Fig.~\ref{fig:channels} come from two tests of a single-qubit channel~\cite{Seddon2019Quantifying,CaoZhao2026Records}. The first asks whether all six vertices of the stabilizer octahedron stay inside it for every exposure time. If they do, an isolated qubit that has lost its magic cannot regain it. The second asks whether the channel also preserves stabilizer mixtures when applied to one part of a state with an untouched reference, a property called complete stabilizer preservation. When it holds at every duration, magic loss is permanent at every register size. Both impossibility statements in Theorem~\ref{thm:min_rebirth_size} follow from these single-qubit channel properties.

In the two-qubit region the channel passes only the first test, so a protected reference can reveal magic generated by the channel. To attain the minimum of two in Theorem~\ref{thm:min_rebirth_size}, both qubits must undergo the noise and the pair must spend a finite interval without magic. A cat with $r$ slightly below one realizes this sequence. Its coherence initially exceeds the population $P_{00}$ of $\ket{00}$ and subsequently exceeds the population $P_{11}$ of $\ket{11}$. Between these events the two populations become equal and both strictly
exceed the coherence, which places the register in the stabilizer polytope for
a finite interval (Methods). Every channel in this family that permits
rebirth therefore permits it with at most two qubits.

For a pair with $a=0.55$, $\vartheta=0.25$ and $r=0.95$ (star in Fig.~\ref{fig:channels}), the second line of Eq.~\eqref{eq:min_rebirth_size} applies. As $\gamma=1-q$ increases, magic first vanishes at $\gamma=0.0747$
and returns at $0.1313$. The excess is $0.0404$ at $\gamma=0.50$.
Entanglement disappears at $\gamma=0.5868$, and at $\gamma=0.60$ the pair
is fully separable with excess $0.0251$. Magic vanishes for a second time
at $\gamma=0.6980$ on the approach to thermal equilibrium.

\begin{figure}[tb]
\centering
\includegraphics[width=.97\linewidth]{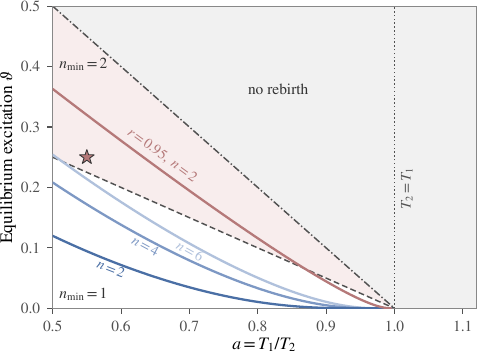}
\caption{\textbf{Minimum register size for magic rebirth under thermal relaxation.}
The dashed line $\vartheta=(1-a)/2$ and the dash-dotted line $\vartheta=1-a$ bound the region where two qubits are necessary and sufficient (pink). One qubit suffices below the dashed line, and rebirth is impossible on or above the dash-dotted line. For each real cat shown, the solid curve gives the numerical threshold below which magic returns. Blue curves use $\alpha=0.4$ at $n=2,4,6$. The rose curve uses $r=0.95$ at $n=2$ and the star marks its trajectory parameters $(a,\vartheta)=(0.55,0.25)$. The dotted line marks $T_2=T_1$.}
\label{fig:channels}
\end{figure}

For a real cat pair in the two-qubit region, magic must return while the pair is still entangled. Once such a pair is both separable and a stabilizer mixture, it has a decomposition into products of single-qubit stabilizer states, and subsequent local channels preserve this convex set (Methods). Magic that returned earlier can persist after entanglement dies.

For a pair with the same population symmetry and a single complex endpoint coherence $c$, stabilizer membership compares $|\operatorname{Re}c|+|\operatorname{Im}c|$ with the endpoint populations, whereas separability compares $|c|$ with the common population $P_{01}=P_{10}$. At the same relaxation rates, a mixed pair with a coherence of phase $\pi/4$ loses and regains magic while remaining separable from preparation onward. Both marginals stay diagonal, and the reborn excess peaks at $3.79\times10^{-4}$. More generally, a two-qubit $X$ state that is separable at the onset of rebirth must satisfy $a<1-(1+1/\sqrt{2})\vartheta$ (Methods). At $\vartheta=1/4$, this confines such returns to the part of the two-qubit region between $a=1/2$ and the bound at $a\simeq0.5732$.

Theorem~\ref{thm:min_rebirth_size} allows the preparation to vary with register size. If the amplitude ratio is held fixed instead, a larger cat can exhibit rebirth under conditions where a smaller one cannot. To see how $n$ enters, let $U=\vartheta(1-q)$ and
$B=\vartheta+(1-\vartheta)q$ be the probabilities that a qubit ends in
$\ket1$ given that it started in $\ket0$ or $\ket1$, respectively. The
coherence competes with the fully excited population according to
\begin{equation}
\label{eq:thermal_fixed_time_cat}
 \frac{c}{P_n}=\frac{r(q^a/B)^n}{1+r^2(U/B)^n}.
\end{equation}
At fixed $0<q<1$, if $q^a>B$, the coherence exceeds $P_n$ for every
sufficiently large register. If $q^a\le B$, it exceeds $P_n$ for none at
that time. The absolute excess may still decrease exponentially with $n$.
At zero temperature, $U=0$ and $B=q$. Rebirth of the fixed real cat then
occurs precisely for
\begin{equation}
 \label{eq:T2T1}
 T_2>T_1,\qquad 1-\gamma_+=r^{1/[n(1-a)]}.
\end{equation}
At positive temperature the coherence tends to zero while
$P_n\to\vartheta^n>0$, so a reborn branch ends in finite time. For fixed
$r$, the critical thermal population $\vartheta_c(a,r,n)$ for cat rebirth
tends to the upper phase boundary,
\begin{equation}
\label{eq:theta_large_n}
 \lim_{n\to\infty}\vartheta_c(a,r,n)=1-a\qquad(1/2\le a<1).
\end{equation}
In Fig.~\ref{fig:channels} the nearly balanced pair ($r=0.95$) reaches farther into the two-qubit region than the three preparations at $\alpha=0.4$. Supplementary Note~\ref{si:thermal} derives the threshold for fixed $r$ and $n$.

\section{Late-time magic for general inputs}
\label{sec:input_dependence}

For the cat, the two strings connected by the coherence differ at every site, and their excitation numbers differ by $n$. In a general input these counts can differ. The strings $01$ and $10$, for example, differ at both sites and have the same excitation number. Let $|u|$ be the number of
excitations in a computational-basis string $u$ and $d_H(u,v)$ the number
of sites on which $u$ and $v$ differ. For an initial density matrix with at
least one nonzero off-diagonal element, define
\begin{equation}
 f(\rho)=\max_{\substack{u\ne v\\ \rho_{uv}\ne0}}
 \frac{\bigl|\,|u|-|v|\,\bigr|}{d_H(u,v)}.
 \label{eq:input_coherence_slope}
\end{equation}
The numerator counts the net change in excitation number, whereas local
dephasing acts on every differing site counted by the denominator. Thus
$0\le f(\rho)\le1$.

\Needspace{7\baselineskip}
\begin{theorem}[Late-time magic]
\label{thm:late_time_magic}
Let a finite register in an arbitrary initial state with at least one
coherence undergo independent zero-temperature relaxation and pure dephasing
with common $T_1,T_2$ in the frame specified above. If $T_2>2T_1/[1+f(\rho)]$, the state has magic at every sufficiently late
finite time. Below this threshold, it remains a stabilizer mixture
from some finite time onward.
\end{theorem}

At positive temperature, the limiting diagonal state has strictly positive populations and lies in the interior of the stabilizer polytope (Supplementary Note~\ref{si:membership}). Every finite register therefore stays inside after a sufficiently long exposure~\cite{CaoZhao2026Records}.

At zero temperature the stationary state is a vertex. At the critical value in Theorem~\ref{thm:late_time_magic} the decay rates alone do not determine the outcome.
For example, at $T_2=T_1$ both $\ket+$ and $(2\ket0+\ket1)/\sqrt5$ have
$f=1$. The former remains in the octahedron, whereas the latter approaches
the ground state from outside it at every finite time. A diagonal input stays diagonal and has zero magic throughout. For the cat, $f=1$ gives the
boundary $T_2=T_1$ in Eq.~\eqref{eq:T2T1}. In the general case, individual
coherences can receive contributions as other qubits decay. The proof in
Methods accounts for this feeding and extends the criterion to unequal
rates.

An earlier interval without magic depends on the preparation. Under pure amplitude
damping, for example, a Haar-random pure input remains outside the stabilizer
polytope at all finite times with probability at least $1-2^{-n}$
(Supplementary Note~\ref{si:inputs}). A state outside the cat family for which an earlier interval
\emph{does} occur is
\begin{equation}
 \frac{\ket{000}+2(\ket{011}+\ket{101}+\ket{110})}{\sqrt{13}}.
 \label{eq:nonx_example}
\end{equation}
This is the unique ground state of a two-local pairing Hamiltonian. Its
coherences connect several basis strings, and pure amplitude damping takes
it into the stabilizer polytope at $\gamma=(\sqrt{5/2}-1)/2$ and back out
immediately after $\gamma=1/2$ (Supplementary Note~\ref{si:inputs}). More generally, even an
arbitrarily small input coherence enters the maximum in
Eq.~\eqref{eq:input_coherence_slope}. A late-time excess can accordingly be
too small to resolve in an experiment.

For a superposition of two computational configurations with positive real amplitudes, a complete test is possible. Under identical pure amplitude damping, a finite interval of magic death followed by rebirth occurs exactly when the excitations of one configuration are a proper subset of those of the other and the more excited configuration has greater initial weight (Supplementary Note~\ref{si:inputs}).

A pure stabilizer input permits an all-time statement under pure amplitude
damping. For each occupied basis string, sum the relaxation rates of
its excited sites. If every occupied string has the same total rate,
the state remains a stabilizer mixture for all times. If the totals
differ, the state has magic at every finite $t>0$ (Methods).
For equal rates the two terms of $(\ket{01}+\ket{10})/\sqrt2$ decay at
the same rate and the state remains a stabilizer mixture. A mismatch in the
rates makes this same input acquire magic immediately. Neither case admits
an intervening interval of magic death.

\section{Entanglement in the register and in its reservoirs}
\label{sec:mirror}

For the real cat under pure amplitude damping, the time at which magic
returns to the register can be read off from the state of the emitted
excitations. Across every bipartition the register has a negative partial
transpose until $\gamma_e=r^{2/n}$ and is fully separable from this value
onward~\cite{Aolita2008Scaling}. Together with Eq.~\eqref{eq:gp_n}, this
gives the threshold identity
\begin{equation}
 \gamma_e+\gamma_+=1.
 \label{eq:complementarity}
\end{equation}
For two qubits the two curves in Fig.~\ref{fig:mechanism}b are reflections about $\gamma=1/2$. As $\alpha$ increases, entanglement death moves from before the interval of zero magic to within it and then to after it (regions I--III).

A qubit coupled to a vacuum reservoir retains its excitation with amplitude
$\sqrt{1-\gamma}$ and emits it with amplitude $\sqrt\gamma$. Interchanging
the qubit with its reservoir therefore exchanges $\gamma$ and $1-\gamma$.
For the whole register the reduced state of the reservoirs is
$\rho_E(\gamma)=\rho_n(1-\gamma)$ in the corresponding bases. It follows
that the boundary $\gamma_+$ for magic in the register also separates full
separability of the reservoirs from entanglement across every bipartition.
For two qubits, L\'opez \emph{et al.} derived the reflected times for
entanglement death in the register and entanglement birth in the
reservoirs~\cite{Lopez2008Sudden}. Here that reservoir birth also fixes the
return of magic, for every register size. In the two-qubit case the relation
extends beyond the initial crossing. At each point of the reborn branch,
$\mathcal R(\rho_2)-1=[(1-\gamma)/\gamma]\,\mathcal C(\rho_E)$, where
$\mathcal C(\rho_E)$ is the concurrence of the two
reservoirs~\cite{Wootters1998Entanglement} (Fig.~\ref{fig:traj}b).

\begin{figure}[!tbp]
\centering
\includegraphics[width=.925\columnwidth]{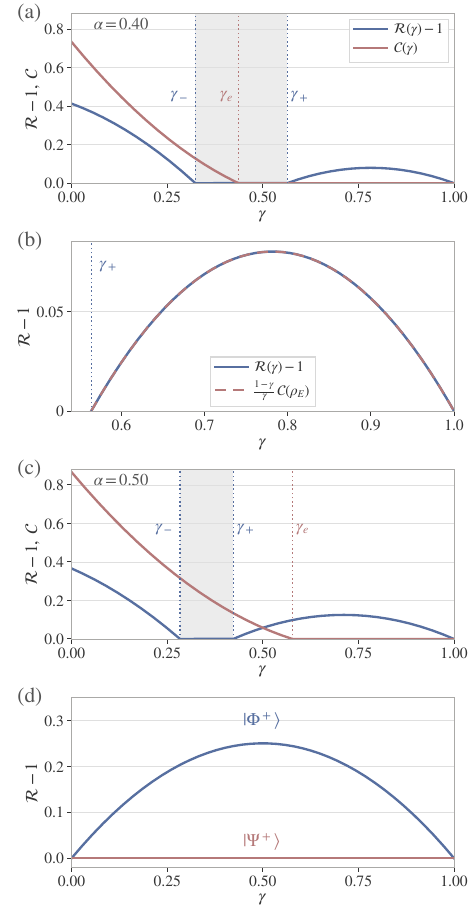}
\caption{\textbf{Magic and entanglement under two-qubit amplitude damping.}
\textbf{a} At $\alpha=0.40$, entanglement (red) dies during the interval without magic (shaded). Blue gives the robustness excess.
\textbf{b} After rebirth the excess equals the reservoir concurrence multiplied by $(1-\gamma)/\gamma$.
\textbf{c} At $\alpha=0.50$, magic returns before entanglement dies.
\textbf{d} Under the same damping, $\ket{\Phi^+}=(\ket{00}+\ket{11})/\sqrt2$ acquires magic immediately and $\ket{\Psi^+}=(\ket{01}+\ket{10})/\sqrt2$ stays a stabilizer mixture.}
\label{fig:traj}
\end{figure}

The ordering of entanglement death and magic return depends on $r$. Figure~\ref{fig:traj}a
shows magic returning after entanglement has already vanished. Figure~\ref{fig:traj}c
uses the amplitude ratio of state~II in the photon experiment of Almeida
\emph{et al.}~\cite{Almeida2007Environment}. For the ideal real input
with $\alpha=1/2$, magic vanishes at $\gamma\simeq0.28$ and returns at
$0.42$, before entanglement dies at $0.58$.
In Fig.~\ref{fig:traj}d both Bell states start with zero magic. Damping creates magic for $\ket{\Phi^+}$,
which superposes different excitation numbers. The state $\ket{\Psi^+}$
has exactly one excitation and stays a stabilizer mixture because both configurations have the same total decay rate.

The whole reborn branch is fully separable precisely when
$\gamma_e\le\gamma_+$, or $r\le2^{-n/2}$, in which case its peak excess
is bounded by $1/[2(1+2^n)]$. Even when the branch starts out entangled,
the excess at any fully separable point on the branch is at most $2^{1-n}$
(Methods). The size-independent peak in Fig.~\ref{fig:deadwindow}b thus
occurs while the register is still entangled for sufficiently large $n$ at
fixed $r$. A relative phase in the input changes the stabilizer condition through
$|\operatorname{Re}c|+|\operatorname{Im}c|$, whereas entanglement depends
on $|c|$. Additional dephasing can also break the threshold identity in Eq.~\eqref{eq:complementarity} (Supplementary Note~\ref{si:mirror}).

\section{Access to the reborn magic}
\label{sec:extract_main}

On the fully separable part of a reborn cat trajectory, the register has magic although all its proper marginals are diagonal. Wei \emph{et al.} call magic absent from each party's marginal \emph{entirely nonlocal} and construct separable states from which the parties can recover it by local stabilizer operations and classical communication~\cite{Wei2026ENM}. For any fully separable state, the product decomposition supplies a local hidden-variable model for arbitrary local measurements. Entangled hypergraph states can likewise have magic absent from their marginals~\cite{Wei2024Noise}. Other definitions ask how much magic survives local basis changes~\cite{Qian2025Quantum,Sierant2026Exact} or shallow local circuits~\cite{Korbany2025Long,Wei2026Long}. For the cat states studied here, extraction depends on which qubits can be operated on in the same laboratory.

The Triangle Criterion characterizes when joint stabilizer operations with postselection can convert one copy of a mixed state into a nonstabilizer qubit~\cite{Liu2025A}. A laboratory holding the whole cat register can transfer its full robustness to
one output qubit with a deterministic protocol.
Its decoder measures the neighboring parities $Z_iZ_{i+1}$ for
$i=1,\ldots,n-1$. When all parities are even, it maps the two configurations
$\ket{0^n}$ and $\ket{1^n}$ to $\ket0$ and $\ket1$. For any other outcome
the output is $\ket+$. Discarding the outcome gives a qubit with Bloch
vector $(2c+1-P_0-P_n,0,P_0-P_n)$ and robustness
\begin{equation}
\label{eq:fixed_cat_decoding}
 \mathcal R_{\mathrm{joint}}=1+2[c-\min(P_0,P_n)]_+=\mathcal R(\rho_n).
\end{equation}
On the reborn branch both nonzero components of the decoded Bloch vector are positive. Applying a Hadamard gate with probability $1/2$ makes them equal to $1/2+c-P_n>1/2$. Independent copies of this output can therefore be distilled by ideal stabilizer operations~\cite{Reichardt2005Quantum}.

If each site holds one qubit, its allowed operations are Clifford gates,
stabilizer preparations and Pauli measurements on its own systems. The
sites may coordinate their choices using classical communication in either
direction, starting from a single copy and no shared ancillary entanglement.
For the real cat with a single coherence $c$, let $p_0,p_1$ denote the
populations of one qubit. The largest excess obtainable on one output
qubit, including postselection, is $2[c-\min(p_0,p_1)]_+$. For pure
amplitude damping this becomes
\begin{equation}
 \label{eq:sitewise_cat_decoding}
 \mathcal R_{\mathrm{site}}-1=2[c-(\alpha^2+\beta^2\gamma)]_+ .
\end{equation}
Measuring $X$ at every other site and sending the parity of the outcomes
to the receiver attains the bound on every run (Methods and Supplementary Note~\ref{si:local}).
On the reborn branch $c<\alpha^2$, so the locally extractable excess is zero even while the register is still entangled. For the two-qubit register
in Fig.~\ref{fig:accessibility}a, the state is fully separable from
$\gamma=2/5$, magic returns just above $\gamma=3/5$, and the peak has
excess $2/29$ at $\gamma=4/5$.

The same conclusion holds on the thermal cat branch with $c>P_n$.
Positivity of the endpoint block gives $p_0\ge P_0>c$, and the one-qubit
excited population is $p_1=\alpha^2U+\beta^2B>c$ for $r<1$ and $n\ge2$
(Supplementary Note~\ref{si:local}). For the parameters marked by the star in Fig.~\ref{fig:channels}, the pair is fully separable at $\gamma=0.60$ and has excess $0.0251$. The joint decoder transfers
all of it to one output qubit, whereas separated sites cannot extract a
nonstabilizer qubit even if they select their measurement outcomes.

One can also divide a single copy between two laboratories holding groups
of $k$ and $h=n-k$ qubits, with no shared ancillary entanglement. Each
laboratory has joint stabilizer control over its own qubits, and the two may
communicate classically. On the reborn branch under pure amplitude damping,
the laboratory with $k$ qubits can extract one nonstabilizer qubit with
nonzero probability precisely when
\begin{equation}
 \label{eq:access_directed_cut}
 c>P_n[1+(\gamma/q)^h],\qquad 1\le h<n.
\end{equation}
Once the state is fully separable, this requires a strict majority at the
receiving laboratory. Every such majority eventually suffices, even as the
available excess tends to zero when $\gamma\to1$. Figure~\ref{fig:accessibility}b
shows how the required control changes during the decay of a six-qubit cat.
With one copy, all six qubits must be in one laboratory at the peak. Later five qubits and
then four suffice.

\begin{figure}[!tbp]
\centering
\includegraphics[width=.92\linewidth]{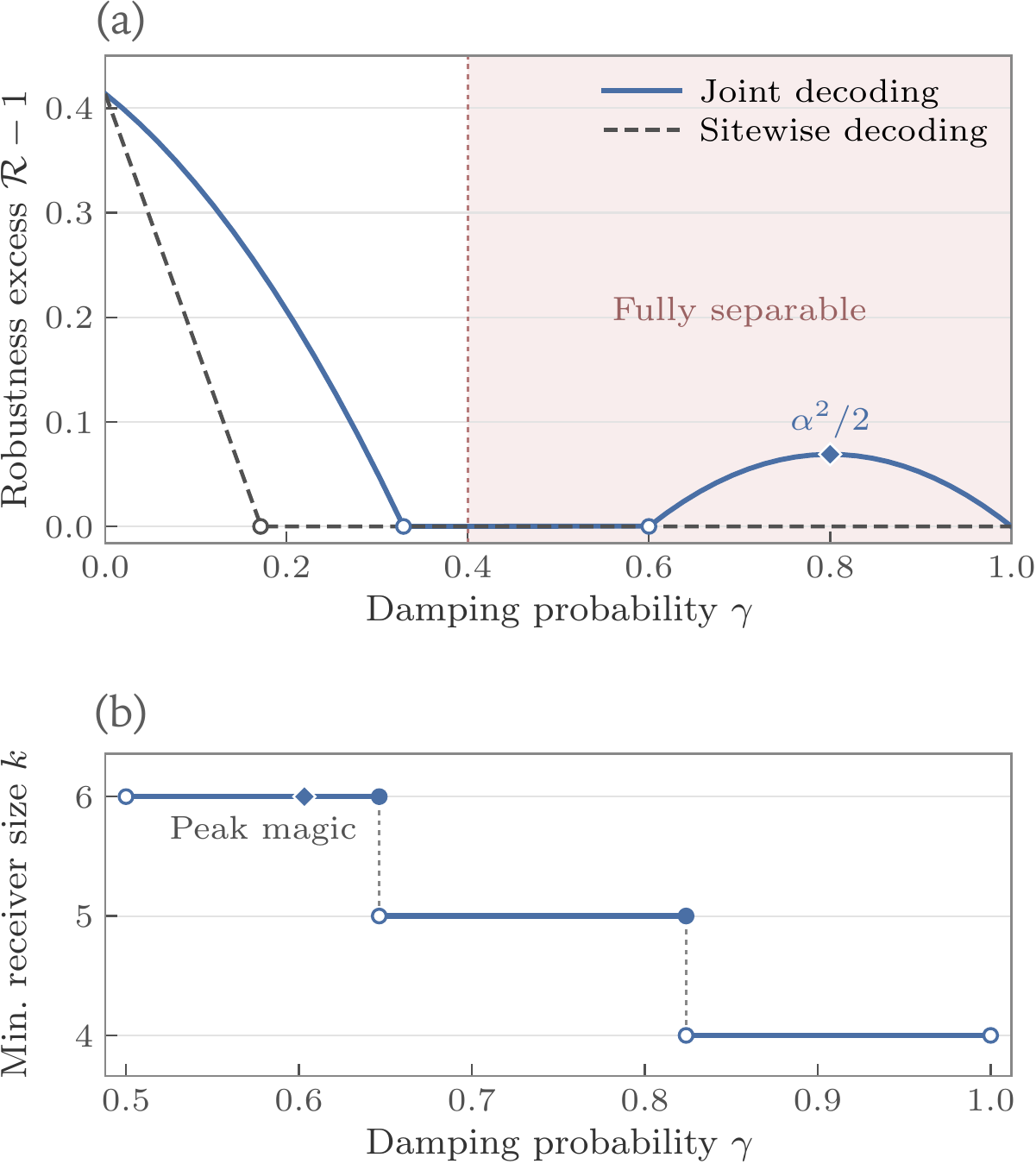}
\caption{\textbf{Single-copy extraction of magic by stabilizer operations.}
\textbf{a} Excess robustness extractable on one qubit with joint control
(solid) or with one qubit at each site and classical communication
(dashed), for $n=2$ and $r=2/5$. The dashed curve is optimal even with
postselection. Shading marks full separability. Circles mark the
thresholds and the diamond the reborn peak. \textbf{b} Smallest number $k$
of qubits that a receiving laboratory must hold to extract one
nonstabilizer qubit with nonzero probability, for $n=6$ and $r=1/8$. The
other $6-k$ qubits can be operated on jointly in a second laboratory. The laboratories share no ancillary entanglement. The register is fully separable along the plotted reborn branch. A
filled endpoint includes the boundary value and an open endpoint excludes
it. At the diamond all six qubits are needed.}
\label{fig:accessibility}
\end{figure}

For a real cat pair on the reborn branch under identical relaxation, two separated sites can extract a nonstabilizer qubit from finitely many copies exactly when the pair is entangled. Each site checks the parities of its qubits across copies, as in recurrence purification~\cite{Bennett1996Purification}. At $\gamma=0.50$ for the star in Fig.~\ref{fig:channels}, five copies suffice. For a fully separable real cat of any even size under identical relaxation, no site can extract a nonstabilizer qubit from any number of copies, even with postselection (Methods and Supplementary Note~\ref{si:local}).

\section{Discussion}
Two equally exposed qubits can regain magic under noise that makes its loss permanent for each qubit alone. The upper boundary for rebirth in Theorem~\ref{thm:min_rebirth_size}, $a=1-\vartheta$, has a second meaning. Below it, relaxing just one qubit of a Bell pair for a suitable duration generates magic with the other qubit protected~\cite{CaoZhao2026Records}. Exposing both qubits leaves the boundary in the same place, even when magic must first vanish for a finite interval. For $a\ge1-\vartheta$, each exposure admits a nonnegative mixture of stabilizer operations, which allows efficient sampling in a circuit with stabilizer preparations and control~\cite{CaoZhao2026Records}. The two-qubit region has width $\min(\vartheta,1/2-\vartheta)$ in $T_1/T_2$. It is widest at $\vartheta=1/4$ and narrows linearly as $\vartheta$ approaches zero. At the bath parameters marked by the star in Fig.~\ref{fig:channels}, a mixed pair with a complex coherence can lose and regain magic while remaining separable throughout. A real cat pair in this region regains magic while still entangled and can retain it after entanglement dies. With pure amplitude damping the order can be reversed. The onset of entanglement in the reservoirs then fixes the return time for every real cat considered here.

Local noise can create quantum discord~\cite{Streltsov2011Discord,Ciccarello2012Discord}. Many-body magic~\cite{Liu2022Many} can spread in random circuits~\cite{Turkeshi2025Magic}, change under measurement~\cite{Niroula2024Phase,Tarabunga2025Magic,Scocco2026Rise}, and increase transiently in an interacting spin chain under dephasing~\cite{Sticlet2025Nonstabilizerness}. In the register studied here, the qubits can regain magic even as the coherence that carries it decays monotonically.

Temperature can also change which resource survives at equilibrium. Sufficiently hot Gibbs states are separable for bounded-degree local spin Hamiltonians with bounded couplings~\cite{Bakshi2024High}, are mixtures of Gaussian states for bounded-degree local fermionic Hamiltonians~\cite{Ramkumar2026Fermionic}, and can be stabilizer mixtures even for long-range Pauli interactions of bounded weight and bounded total strength per site~\cite{Putterman2026Thermal}. Lyu \emph{et al.} found that magic survives to higher temperatures than entanglement in the transverse-field Ising model, whereas after a generic spin-chain quench the magic of a subregion typically dies first~\cite{Lyu2026Fate}. Here the equilibrium state is a diagonal product at every temperature. At fixed temperature and preparation, changing the ratio of decay rates can make the trajectory leave the stabilizer polytope on its way to equilibrium.

Almeida \emph{et al.} applied equal, tunable damping to the two photons of a polarization-entangled pair and reconstructed their joint state at each setting~\cite{Almeida2007Environment}. Their state~II had the amplitude ratio used in Fig.~\ref{fig:traj}c. For the ideal real input the reborn excess peaks at $1/8$ near $\gamma=0.71$, where the pair is already separable and $c-P_{11}=1/16$. For a pure state with this amplitude ratio, a finite interval without magic requires the relative phase to lie less than $\tfrac{1}{2}\arcsin(1/3)\simeq9.74^\circ$ from a multiple of $\pi/2$ in the chosen Pauli frame (Supplementary Note~\ref{si:mirror}). The experiment reported purities of about $0.91$--$0.97$ for its two inputs and did not specify the phase of state~II. The measured density matrices in a calibrated Pauli frame would allow the crossings to be tested for the actual preparation.

For a superconducting device, Lacroix \emph{et al.} reported relaxation times of $16.7$ and $14.8\,\mu\mathrm{s}$, Ramsey times of $27.2$ and $24.5\,\mu\mathrm{s}$, and excited-state populations of $0.3\%$ and $0.04\%$ for the transmons B1 and B2, respectively~\cite{Lacroix2020Continuous}. Using these values in an independent exponential-decay model with $T_2=T_2^*$, the same cat as in Fig.~\ref{fig:traj}a ($\alpha=0.4$) loses magic at $5.7\,\mu\mathrm{s}$, becomes separable at $7.8\,\mu\mathrm{s}$ and regains magic at $16.7\,\mu\mathrm{s}$. The predicted excess peaks at $0.036$ at $26.7\,\mu\mathrm{s}$, where $c-P_{11}=0.018$. The reported probabilities of correctly assigning $\ket1$ were $97.70\%$ and $94.45\%$ for B1 and B2. Resolving this small margin would require calibrated joint readout and an uncertainty estimate for $c-P_{11}$. Both fitted channels satisfy $T_1/T_2<1-2\vartheta$, so either qubit could exhibit rebirth on its own. The estimate gives the time and signal scales for observing the population crossings. On the same device, A3, A4 and B3 had $T_2^*<T_1$. In this regime their fitted channels are mixtures of Clifford operations and resets~\cite{Garner2026Thermal}, so a register subject only to such channels cannot undergo rebirth. Resolving the two-qubit region requires control of both the decay rates and the thermal population. For example, at $\vartheta=0.1$ it occupies $1/0.9<T_2/T_1\le 1/0.8$.

Randomized measurements have been used to estimate magic on processors~\cite{Oliviero2022Measuring}, and witnesses are available for mixed states~\cite{Haug2026Efficient,Warmuz2025Magic}. For an ideal real cat pair, two local Pauli settings suffice to test the population crossings. Computational-basis counts estimate both endpoint populations, and $\langle X\otimes X\rangle$ gives twice the coherence. By comparing their records, separated sites can thus certify shared magic even when they cannot extract a nonstabilizer qubit from one copy by local stabilizer operations. If other coherences are present, a phase scan isolates the endpoint coherence~\cite{Monz2011Fourteen}, and a violation of the population bound still certifies magic. Establishing zero magic requires a membership test. For general mixed states even approximate membership is computationally hard under the exponential time hypothesis~\cite{Leone2026Unbearable}. Tomography gives a sufficient test at each sampled time if, within statistical uncertainty, each population exceeds the sum of the absolute real and imaginary parts of the off-diagonal entries in its row (Methods).

It remains to determine which other local Markovian evolutions allow a pair to regain magic when its loss is permanent for every isolated qubit. For a real cat register with an odd number of sites ($n\ge3$), it is also open whether separated sites can ever extract a nonstabilizer qubit from finitely many copies on the fully separable part of the reborn branch.

\section{Methods}

\subsection{Stabilizer membership and the duration of magic death}
\label{methods:cat}

A pure stabilizer state's nonzero computational-basis amplitudes all
have the same modulus~\cite{Dehaene2003Clifford}. Hence for every
stabilizer mixture \(\sigma\) and any two basis strings \(u,v\),
\begin{equation}
 |\sigma_{uv}|\leq\min(\sigma_{uu},\sigma_{vv}).
 \label{eq:pair_obstruction}
\end{equation}
On the real cat trajectory the only off-diagonal pair is
\(\rho_{0^n,1^n}=c\geq0\). If \(c\leq\min(P_0,P_n)\), subtracting
\(c(\ket{0^n}+\ket{1^n})(\bra{0^n}+\bra{1^n})\) from the state
leaves a nonnegative diagonal operator. This gives the stated
membership test. A matching signed decomposition and a dual witness give Eq.~\eqref{eq:RoM_GHZ} (Supplementary Note~\ref{si:membership}). It is also the single-coherence case of the robustness formula for generalized $X$ states, whose computational-basis entries are confined to the diagonal and antidiagonal~\cite{deOliveira2026Symmetry}.

For the pure cat input, the difference
\(P_0-c=\alpha^2+\beta^2\gamma^n-\alpha\beta(1-\gamma)^{n/2}\)
increases strictly from a negative value at \(\gamma=0\) to
\(\beta^2\gamma_+^n>0\) at \(\gamma_+\). Thus it has a unique
zero \(\gamma_-\). Over \([\gamma_-,\gamma_+]\) its derivative
is at least \((n/2)\alpha\beta\,r^{1-2/n}\) for \(n\geq2\). The
mean-value theorem gives the bound plotted in Fig.~\ref{fig:deadwindow}a,
\begin{equation}
 \gamma_+-\gamma_-\ \leq\
 \frac{2}{n}\,r^{2/n-2}(1-r^{2/n})^n .
 \label{eq:deadwindow_bound}
\end{equation}
At fixed \(r\in(0,1)\), \(1-r^{2/n}\leq2\ln(1/r)/n\), giving
the faster-than-exponential decay.

For an input of visibility \(\mu\) as in Eq.~\eqref{eq:visibility_window},
allow each qubit to have its own relaxation time. Write
\(e^{-x}=\prod_i(1-\gamma_i)\) and \(G=\prod_i\gamma_i\).
The endpoint populations and the coherence are
\[
 P_0=\alpha^2+\beta^2G,\qquad P_n=\beta^2 e^{-x},\qquad
 c=\mu\alpha\beta e^{-x/2}.
\]
At \(x\geq2\ln(\mu/r)\) we have \(c\leq\alpha^2\leq P_0\).
At \(x\leq x_+=2\ln[1/(\mu r)]\) we also have \(c\leq P_n\).
Beyond \(x_+\), \(P_n<c\) at every finite \(x\). Maximizing
\(2\beta e^{-x/2}(\mu\alpha-\beta e^{-x/2})\) over \(x>x_+\)
gives a peak of \(\mu^2\alpha^2/2\) at \(e^{-x/2}=\mu r/2\).
For fixed \(x\), the arithmetic--geometric mean inequality applied
to the \(\gamma_i\), and then to the \(1-\gamma_i\), yields
\(G\leq(1-e^{-x/n})^n\). The right-hand side tends to zero
uniformly on bounded \(x\)-intervals. Since \(P_0-c\) increases
strictly along the semigroup, its root tends to \(2\ln(\mu/r)\).
The actual zero-magic width in total exposure tends to
\(4\ln(1/\mu)\) for fixed \(r<\mu<1\).

\subsection{Proof of the minimum register size}
\label{sec:thermal_proof}

Write \(z_{\rm eq}=1-2\vartheta\) for the equilibrium polarization and \(\lambda=q^a\). The single-qubit
map acts on Bloch coordinates as
\[
 (m_x,m_y,m_z)\longmapsto(\lambda m_x,\lambda m_y,q m_z+z_{\rm eq}(1-q)).
\]
Testing the six vertices of the stabilizer octahedron shows that it
preserves the one-qubit stabilizer polytope if and only if
\(\lambda\le 1-z_{\rm eq}(1-q)\). It preserves stabilizer states in the
presence of any additional reference qubits if and only if
\(\lambda\le B\), where \(B=\vartheta+(1-\vartheta)q\).
These are the two channel criteria in Supplementary Note S1 of
Ref.~\cite{CaoZhao2026Records}. The second can also be read off
from the channel's normalized Choi state~\cite{Seddon2019Quantifying}.
Its sole coherence is \(\lambda/2\), and its endpoint populations
are \((1-U)/2\) and \(B/2\), with \(1-U\ge B\).

For \(0\le b\le1\), the inequality \(q^a\le 1-b+bq\) holds for
all \(q\in(0,1]\) exactly when \(a\ge b\). Sufficiency follows from
\(q^a\le q^b\le 1-b+bq\), and necessity follows by expanding about
\(q=1\). Applying this to \(b=z_{\rm eq}\) and \(b=1-\vartheta\) gives the
two boundaries in Eq.~\eqref{eq:min_rebirth_size}. In a semigroup
the map advancing the state from \(t\) to \(t+u\) is the channel
for duration \(u\). Thus the one-qubit preserving condition for
all durations makes a one-qubit death permanent. The condition with
reference qubits has the same consequence for every register size,
since each local channel can then be applied in turn while preserving
the stabilizer polytope of the entire register.

For one-qubit attainability take \(a<z_{\rm eq}\). Starting at \(\ket{+}\)
gives \(m_x(t)+m_z(t)=1+(z_{\rm eq}-a)t/T_1+O((t/T_1)^2)\). Choose a small
positive time \(t_*\) for which \(m_x(t_*)+m_z(t_*)>1\). Perturb the
initial pure Bloch vector to
\((\sqrt{1-\epsilon^2},0,-\epsilon)\), where \(0<\epsilon\ll1\).
This input has magic. By continuity, \(m_z(t_*)>0\) and
\(m_x(t_*)+m_z(t_*)>1\) still hold for small enough \(\epsilon\).
At the intervening time when \(m_z=0\), one has
\(0<m_x=q^a\sqrt{1-\epsilon^2}<1\). The state is strictly inside
the octahedron and stays in it for an open interval.

For the pair, start from
\((r\ket{00}+\ket{11})/\sqrt{1+r^2}\). More generally the
endpoint populations and coherence for \(n\) qubits are
\begin{align}
 (1+r^2)P_0 &=r^2 A^n+D^n,\nonumber\\
 (1+r^2)P_n &=r^2 U^n+B^n,\nonumber\\
 (1+r^2)c   &=r q^{an},
 \label{eq:thermal_cat_endpoints}
\end{align}
where \(A=1-U\) and \(D=1-B\).
For \(n=2,\ r=1\), expansion at short times gives
\[
 c-P_{11}=\tfrac12(q^{2a}-B^2-U^2)
       =(1-\vartheta-a)t/T_1+O((t/T_1)^2).
\]
If \(a<1-\vartheta\), fix a time \(t_*>0\) at which this margin is
positive. By continuity it remains positive with \(r<1\) sufficiently
close to one. At \(t=0\) such a state satisfies \(P_{00}<c<P_{11}\),
so it initially has magic. At \(t_*\), positivity of the endpoint
block gives \(P_{00}>P_{11}\). At an intermediate time \(P_{00}=P_{11}\).
The endpoint block is strictly positive definite for every
\(0<q<1\). Indeed \(AB=q+\vartheta(1-\vartheta)(1-q)^2
\ge q\ge q^{2a}\), and Eq.~\eqref{eq:thermal_cat_endpoints} gives
\[
 (1+r^2)^2(P_0P_n-c^2)
 \ge r^2\big[(AB)^n-q^{2an}\big]+(DB)^n>0 .
\]
Consequently \(c<P_{00}=P_{11}\) at the intermediate time. The two
strict inequalities persist on an open interval. They are sufficient
for stabilizer membership because the real endpoint coherence is
the only off-diagonal pair. This proves two-qubit attainability.

The same strict-positivity argument applies at any \(n\) if
\(c>P_n\) is found at some time after starting with \(r<1\).
Equation~\eqref{eq:thermal_fixed_time_cat} then gives a simple test
at a prescribed time. Its denominator lies between \(1\) and
\(1+r^2\), because \(0\le U<B\) for \(0<q<1\).
Hence \(c>P_n\) for all sufficiently large \(n\) exactly when
\(q^a>B\). For \(\vartheta\ge1-a\) this inequality fails at all
times. For \(\vartheta<1-a\) it holds for \(q<1\) sufficiently
close to one. At fixed \(q<1\), increasing \(\vartheta\) increases
both \(U\) and \(B\) and leaves \(c\) fixed. These observations give
Eq.~\eqref{eq:theta_large_n}. When \(\vartheta=0\), one has
\(c/P_n=r q^{-(1-a)n}\), yielding Eq.~\eqref{eq:T2T1}.

For the two-qubit numerical example, the two remaining populations
are \(P_{01}=P_{10}=(r^2AU+DB)/(1+r^2)\).
The partial transpose is positive exactly when \(P_{01}\ge c\).
Solving \(P_{00}=c\), \(P_{11}=c\), and \(P_{01}=c\) gives, in time
order, the first magic death at \(q=0.9253123\), its return at
\(q=0.8687339\), entanglement death at \(q=0.4132421\), and the
second magic death at \(q=0.3020279\).

\paragraph{Entanglement at the return of magic in a pair.}
For a real cat pair, $c\ge0$ and $P_{01}=P_{10}$. At a point that is both separable and free of magic, positivity of the partial transpose and the endpoint test give $P_{ij}\ge c$ for all four strings. Write $\ket{u_\phi}=(\ket0+e^{i\phi}\ket1)/\sqrt2$ and average the projectors onto $\ket{u_\phi}\ket{u_{-\phi}}$ over $\phi=0,\pi/2,\pi,3\pi/2$. The resulting mixture $\omega$ of product stabilizer states has diagonal entries $1/4$ and a single coherence $\omega_{00,11}=1/4$. Subtracting $4c\,\omega$ leaves a nonnegative diagonal operator. At $a\ge1-2\vartheta$ the local channels preserve such mixtures at every duration, so this pair cannot subsequently regain magic.

For a single complex coherence $c=\rho_{00,11}$, the state is a stabilizer mixture exactly when $|\operatorname{Re}c|+|\operatorname{Im}c|\le\min(P_{00},P_{11})$ (Supplementary Note~\ref{si:membership}). At $a=0.55$ and $\vartheta=1/4$, take populations $(0.289,0.210,0.210,0.291)$ and $c(0)=0.1454(1+i)$. The least eigenvalue of the partial transpose at preparation is $0.210-0.1454\sqrt{2}>0$, so the local channels keep the pair separable throughout. Here $c=c(0)q^{1.1}$, $P_{11}=0.0625+0.1255q+0.103q^2$ and $P_{00}=P_{11}+0.5-0.502q$. The state loses magic at $\gamma\simeq0.00367$, regains it at $0.02029$ and loses it again at $0.11418$. The reborn excess $2(0.2908q^{1.1}-P_{11})$ peaks at $\gamma\simeq0.06717$.

For $c\ne0$, the phase factor $F=(|\operatorname{Re}c|+|\operatorname{Im}c|)/|c|$ is constant during relaxation. With $\tau=t/T_1$, $F|c|$ decays at rate $2a$, while $P_{00}$ has outflow rate $2\vartheta\le2a$ and nonnegative inflow. Hence $F|c|\le P_{00}$ is preserved. For the other endpoint, $\dot P_{11}=\vartheta(P_{01}+P_{10})-2(1-\vartheta)P_{11}$. Once the pair is separable it remains so, and positivity of the partial transpose gives $P_{01}+P_{10}\ge2\sqrt{P_{01}P_{10}}\ge2|c|$. Consequently,
\[
 \begin{aligned}
 &\frac{d}{d\tau}\!\left[e^{2(1-\vartheta)\tau}(F|c|-P_{11})\right]\\
 &\qquad\le 2e^{2(1-\vartheta)\tau}F|c|\,[1-(1+1/F)\vartheta-a].
 \end{aligned}
\]
If the bracket is nonpositive, $F|c|$ cannot overtake $P_{11}$ after magic death while the pair is separable. This gives the strict necessary condition $a<1-(1+1/F)\vartheta$. For a general two-qubit $X$ state, its additional $01$--$10$ coherence decays at rate $2a\ge1$, whereas $P_{01}$ and $P_{10}$ each have unit outflow rate and nonnegative inflow. Once satisfied, the stabilizer condition for that block therefore persists (Supplementary Note~\ref{si:membership}). If $c=0$ this rules out a return altogether. Otherwise $F\le\sqrt2$ gives the phase-independent bound stated after Theorem~\ref{thm:min_rebirth_size}.

\subsection{Proof of the late-time criterion}
\label{sec:general_input_methods}

Identify a computational-basis string with its set of excited sites.
Write $\rho_{uv}=\langle u|\rho(0)|v\rangle$,
$w_u=\sum_{i\in u}\kappa_i$ and $d_i=1/T_{2,i}-\kappa_i/2\ge0$, where
$\kappa_i=1/T_{1,i}>0$. Suppose the initial state has a nonzero
off-diagonal entry and define
\begin{equation}
 E_*=\max_{\substack{u\ne v\\ \rho_{uv}\ne0}}
 \left\{\frac{|w_v-w_u|}{2}-\sum_{i\in u\triangle v}d_i\right\}.
 \label{eq:input_coherence_rate}
\end{equation}
The symbol $\triangle$ denotes symmetric difference. For fixed finite
$n$, the state has magic at every sufficiently late finite time if
$E_*>0$ and belongs to the stabilizer polytope at every sufficiently
late time if $E_*<0$. For common rates these are the two strict
inequalities following Eq.~\eqref{eq:input_coherence_slope}.

For the proof set $q_i=e^{-\kappa_it}$ and
$\Gamma_{uv}=(w_u+w_v)/2+\sum_{i\in u\triangle v}d_i$.
On an off-diagonal single-qubit matrix element the channel multiplies
the coherence by $e^{-t/T_{2,i}}$. An excitation present on both sides
of a matrix element can decay on both sides together. Tensoring the
single-qubit maps thus gives
\begin{equation}
 \rho_{uv}(t)= e^{-\Gamma_{uv}t}
 \sum_{J\subseteq (u\cup v)^c}
       \rho_{u\cup J,\,v\cup J}\prod_{j\in J}(1-q_j).
 \label{eq:input_coherence_evolution}
\end{equation}
This also holds for $u=v$. In particular,
$e^{w_v t}\rho_{vv}(t)\to N_v=\sum_{z\supseteq v}\rho_{zz}(0)$.

If $E_*>0$, select a pair that attains this value, oriented so that
$w_v>w_u$. Among all nonzero input coherences with the same two
differences $u\setminus v$ and $v\setminus u$, choose one whose set
of common excitations is maximal under inclusion. This choice preserves
$w_v-w_u$ and $\Gamma_{uv}-(w_u+w_v)/2$. Only the $J=\varnothing$ term in
Eq.~\eqref{eq:input_coherence_evolution} contributes to the selected
entry. Since $0<\rho_{vv}(0)\le N_v$,
\begin{equation}
 \frac{|\rho_{uv}(t)|}{\rho_{vv}(t)}
 \ \ge\ \frac{|\rho_{uv}|}{N_v}e^{(w_v-\Gamma_{uv})t}
 =\frac{|\rho_{uv}|}{N_v}e^{E_*t}.
\end{equation}
Every pure stabilizer state has equal-modulus amplitudes on its
computational support~\cite{Dehaene2003Clifford}. Convexity then gives
$|\sigma_{uv}|\le\min(\sigma_{uu},\sigma_{vv})$ for every stabilizer
mixture $\sigma$. Once the ratio exceeds one the output has magic.

If $E_*<0$, consider any off-diagonal entry for which a term on the
right of Eq.~\eqref{eq:input_coherence_evolution} can contribute. The
input pair in that term has the same directed differences as $(u,v)$.
It follows that $\max(w_u,w_v)-\Gamma_{uv}\le E_*<0$. Positivity of the
input gives $N_u,N_v>0$. The finite sum in
Eq.~\eqref{eq:input_coherence_evolution} is bounded, so
$|\rho_{uv}(t)|/\rho_{uu}(t)$ and $|\rho_{uv}(t)|/\rho_{vv}(t)$ both tend to zero.
There are finitely many entries. Eventually every row satisfies
\begin{equation}
 \rho_{uu}(t)\ \ge\
 \sum_{v\ne u}\bigl(|\operatorname{Re}\rho_{uv}(t)|
                   +|\operatorname{Im}\rho_{uv}(t)|\bigr).
 \label{eq:input_diagonal_dominance}
\end{equation}
For each unordered pair $u,v$, reproduce the real and imaginary parts
of its coherence with projectors onto
$(\ket{u}\pm\ket{v})/\sqrt2$ and $(\ket{u}\pm i\ket{v})/\sqrt2$,
respectively. These are stabilizer states. Their contribution to each
of the two diagonal entries is
$|\operatorname{Re}\rho_{uv}(t)|+|\operatorname{Im}\rho_{uv}(t)|$.
Equation~\eqref{eq:input_diagonal_dominance} leaves nonnegative
residual populations for the computational-basis projectors. This
constructs a stabilizer decomposition. At $E_*=0$ the rates alone do
not settle membership. For example, at $T_2=T_1$ the input $\ket+$
remains a stabilizer mixture, whereas
$(2\ket0+\ket1)/\sqrt5$ has $|\rho_{01}(t)|=2\rho_{11}(t)>0$ at every
finite time.

For a pure stabilizer input $\ket\phi$ under pure relaxation,
set $d_i=0$ and let $K=\sum_i\kappa_i\ket{1}_i\!\bra{1}$ be the total
decay-rate operator. If $\operatorname{Var}_{\phi}(K)=0$, all occupied strings
have the same $w_u$. Conditional on any particular subset $J$ of
downward jumps, the surviving strings each acquire the same amplitude
factor. The normalized branch can therefore be prepared by projecting the
sites in $J$ onto $\ket1$ and flipping those bits to $\ket0$, both
stabilizer operations. If the variance is positive, choose an occupied
string $v$ of maximal rate $w_v$ and an occupied $u$ with $w_u<w_v$.
No occupied string lies strictly above $v$. Equation~\eqref{eq:input_coherence_evolution}
and equal initial amplitude moduli give
$|\rho_{uv}(t)|/\rho_{vv}(t)=e^{(w_v-w_u)t/2}>1$ for every finite $t>0$.

The exact death and return times for the three-qubit state in Eq.~\eqref{eq:nonx_example}, and a two-local Hamiltonian that prepares it as a unique ground state, are derived in Supplementary Note~\ref{si:inputs}.

\subsection{Full separability and the emitted excitations}
\label{methods:mirror}

For the pure input the threshold follows from the multiqubit amplitude-damping result of Aolita \emph{et al.}~\cite{Aolita2008Scaling}. We give a derivation that also covers reduced preparation visibility. For a cut with \(m\) qubits on one side, partial transposition moves
the cat coherence into the subspace spanned by
\(\ket{1^m0^{n-m}}\) and \(\ket{0^m1^{n-m}}\), after reordering
the sites. The other eigenvalues are nonnegative. For an initial
coherence \(\mu\alpha\beta\), including \(\mu=1\), the determinant
of the \(2\times2\) block is
\begin{equation}
 \beta^2 q^n\bigl(\beta^2\gamma^n-\mu^2\alpha^2\bigr).
 \label{eq:cat_pt_visibility}
\end{equation}
The block has a negative eigenvalue up to
\(g=(\mu r)^{2/n}\), independently of the cut.
Full separability criteria are known for the GHZ-diagonal family of D\"ur and Cirac~\cite{Dur2000Classification} and for general $X$ states~\cite{Ha2019Separability}. Here the threshold admits a direct product decomposition. At \(g\), take a product of
single-qubit states \(\sqrt g\ket0+
e^{i\theta_j}\sqrt{1-g}\ket1\). Average the product projectors over
independent uniform phases \(\theta_1,\ldots,\theta_{n-1}\) with
\(\theta_n=-\sum_{j<n}\theta_j\), and denote the resulting state
by \(\tau_g\). Its diagonal is the product distribution with
excitation probability \(1-g\) at every site. Its only coherence
is \([g(1-g)]^{n/2}\) between the two endpoints. Since
\(\beta^2g^{n/2}=\mu\alpha\beta\), the register state at \(g\) is
\(\alpha^2\ket{0^n}\!\bra{0^n}+\beta^2\tau_g\). This gives an
explicit fully separable decomposition. Later local damping preserves it.
For a mixed state, genuine multipartite entanglement has a separate
threshold~\cite{Rafsanjani2012Genuinely}. For example, the pure cat
input loses genuine multipartite entanglement at
\(\gamma=[r/(2^{n-1}-1)]^{2/n}\), earlier than full separability
for \(n\geq3\).

In the vacuum dilation, each initially empty reservoir stays empty
when its qubit is in \(\ket0\). An initial excitation evolves as
\(\ket{1}_S\ket{0}_E\mapsto\sqrt q\ket{1}_S\ket{0}_E+
\sqrt\gamma\ket{0}_S\ket{1}_E\). Tracing out the qubit therefore
applies amplitude damping of strength \(1-\gamma\) to the input in
the reservoir basis. Tensoring the isometries gives
\(\rho_E(\gamma)=\rho_S(1-\gamma)\) for any register input, and
Eq.~\eqref{eq:cat_pt_visibility} gives the reflected thresholds for
the real cat. For \(n=2,\mu=1\), the reservoir concurrence is
\(\mathcal C(\rho_E)=2\beta\gamma[\alpha-\beta(1-\gamma)]_+\).
Multiplying by \((1-\gamma)/\gamma\) reproduces the register's
robustness excess throughout the reborn branch.

For a fully separable point on that branch, \(\mu=1\) and the
partial-transpose condition gives \(r\leq\gamma^{n/2}\). Therefore
\[
 \mathcal R-1 =
 \frac{2q^{n/2}(r-q^{n/2})}{1+r^2}
 \leq 2[\gamma(1-\gamma)]^{n/2}
 \leq 2^{1-n}.
\]
For unequal local damping probabilities the same argument gives
\(\mathcal R-1\leq2\prod_i\sqrt{\gamma_i(1-\gamma_i)}\leq2^{1-n}\)
at fully separable points on the reborn branch (Supplementary Note~\ref{si:mirror}).

\subsection{Measuring the endpoint competition}
\label{methods:cat_measurement}

Computational-basis counts estimate \(P_0\) and \(P_n\) in the same
measurement setting. On the modeled cat trajectory
\(\langle X^{\otimes n}\rangle=2c\). With possible additional
coherences one can instead measure
\(X_\phi^{\otimes n}\), where \(X_\phi=X\cos\phi+Y\sin\phi\).
The Fourier coefficient of \(\langle X_\phi^{\otimes n}\rangle\)
at frequency \(+n\) is \(c=\langle0^n|\rho|1^n\rangle\), since
each complementary pair with Hamming weights \(k,n-k\)
contributes harmonics of order \(\pm(n-2k)\). Thus the highest
harmonic isolates both quadratures of the endpoint coherence~\cite{Monz2011Fourteen}.
Inequality~\eqref{eq:pair_obstruction} makes an observed violation
a magic witness even if the state has preparation errors.
On the specified cat family, estimating both populations and the coherence suffices to test membership. For a general input, Eq.~\eqref{eq:input_diagonal_dominance} gives a sufficient certificate when a lower confidence bound on every diagonal entry exceeds a simultaneous upper bound on the sum of coherences in that row. This certifies each tested time. Extending the certificate between sampled times requires a bound on the intervening evolution.

\subsection{Stabilizer readout}
\label{sec:access_methods}

\paragraph{Joint control.}
The all-even parity outcome has probability $P_0+P_n$ and gives
an unnormalized logical qubit with diagonal entries $P_0,P_n$ and
off-diagonal entry $c$. All other outcomes have total probability
$1-P_0-P_n$. Resetting these outcomes to $\ket{+}$ adds their full
weight to the $X$ component of the output Bloch vector.
For a normalized qubit with Bloch vector $\bm m$ the robustness is
$\max(1,\|\bm m\|_1)$, which gives Eq.~\eqref{eq:fixed_cat_decoding}
for $c\ge0$. Supplementary Note~\ref{si:conditional} gives a sufficient error budget for an imperfect implementation of this decoder.

For a known generalized $X$ state, measuring the same parities and
decoding produces logical qubits with Bloch vectors $\bm m_s$ and
probabilities $p_s$. The sector identity gives the input robustness
as $\sum_s p_s\max(1,\|\bm m_s\|_1)$. For each block with
$\|\bm m_s\|_1>1$, choose a Clifford rotation sending its vector into
the positive octant. Such a rotation exists for every octant since
the Clifford rotations act transitively on the eight faces of the
stabilizer octahedron. Reset all other blocks to $\ket 0$. After
erasing the syndrome, the sum of the three Bloch components is the
input robustness, so the output has that same robustness. The choice
of Clifford and of which blocks to reset can depend on the input.
The cat protocol above uses a single preset rule for its entire trajectory.

\paragraph{One qubit per site.}
Fix a complete measurement record, including any local random choices. Resolve every discarded output by a computational-basis measurement. This refines the original protocol into branches, each described by a product of local Kraus operators. Classical communication chooses these operators, so the factorization holds for each fixed record even if the parties communicate in both directions. Contracting a branch of a stabilizer circuit with Bell pairs on its inputs gives a pure stabilizer Choi vector up to a scalar. Each assisting site thus contributes a bra proportional to a Pauli eigenbra. A $Z$ bra kills the endpoint coherence. The other possibility is that every assisting site contributes an $X$ or $Y$ eigenbra. Such a bra has squared overlap $1/2$ with either basis value. On normalization these factors cancel on every diagonal entry, and the sole coherence acquires a known power of $i$. For the real cat, a local Clifford removes this phase. The receiver is left with coherence $c$ and its original marginal populations $p_0,p_1$.

The receiver's remaining one-qubit operation has a pure two-qubit stabilizer Choi vector. If this vector is a product, the output is a fixed stabilizer state. Otherwise it is locally Clifford-equivalent to a Bell pair and implements a Clifford up to a scalar. Consequently the largest robustness excess on any branch is $2[c-\min(p_0,p_1)]_+$. Coarse-graining branches preserves the upper bound by convexity. Measuring $X$ at all assisting sites and correcting the outcome parity attains it on every run.

For pure damping, $p_0=\alpha^2+\beta^2\gamma$ and $p_1=\beta^2 q\ge c$, giving Eq.~\eqref{eq:sitewise_cat_decoding}. At finite temperature $p_1=\vartheta+(\beta^2-\vartheta)q$ and $c=\alpha\beta q^{an}<p_1$ for $n\ge2$ and $0<r<1$. At global magic death $p_0\ge P_0\ge c$. Subsequently $p_0$ increases and $c$ decreases, so separated sites cannot extract a nonstabilizer qubit from one copy on the reborn branch. Supplementary Note~\ref{si:local} gives the full branch calculation and the extension to two laboratories with joint control within each laboratory.

\paragraph{Two cooperating laboratories.}
The two laboratories share no ancillary entanglement. For the postselected task in Eq.~\eqref{eq:access_directed_cut},
let the assisting laboratory hold $h=n-k$ qubits. Its refined effect
is proportional to a pure stabilizer projection $\ket s\!\bra s$.
The computational-basis amplitudes on the support of $\ket s$ have
equal magnitude and relative phases that are powers of $i$
\cite{Dehaene2003Clifford}. If either $0^h$ or $1^h$ lies outside
this support, the receiving laboratory obtains a diagonal mixture.
Otherwise, divide its unnormalized conditional state by the common
endpoint probability of the effect. Up to a Clifford phase
correction the remaining endpoint coherence is $c$. The all-zero
population is at least $\alpha^2>c$ on the open reborn branch. The
all-one population is at least
$\beta^2 q^k(q^h+\gamma^h)$, since both $0^h$ and $1^h$ contribute
to it. When this second lower bound is at least $c$, the
single-coherence membership criterion puts the conditional receiver
state in the stabilizer polytope. Subsequent local stabilizer
operations preserve membership on each outcome. Conversely, the
assisting laboratory can project onto $(\ket{0^h}+\ket{1^h})/\sqrt2$, which
attains the lower bound exactly. If $c$ exceeds it, joint parity
decoding in the receiver gives a nonstabilizer qubit. This proves
Eq.~\eqref{eq:access_directed_cut} for arbitrary classical
communication in either direction.

For a fully separable reborn state,
$q^{n/2}<r\le\gamma^{n/2}$, so $\gamma>q$.
If $h\ge n/2$, then
$q^{n/2}[1+(\gamma/q)^h]\ge q^{n/2}+\gamma^{n/2}>r$,
which rules out extraction. For $h<n/2$ this left-hand side tends
to zero as $q\to0$. In Fig.~\ref{fig:accessibility}b the
six-qubit branch begins at $\gamma=1/2$. A five-qubit receiver
first succeeds above $\gamma=1-\sqrt{1/8}$, and a four-qubit receiver
first succeeds above $\gamma\simeq0.823899$, where
$q(q^2+\gamma^2)=1/8$. At the peak
$\gamma=1-16^{-1/3}$ both tests fail.

\paragraph{Many copies shared between sites.}
Each site holds one qubit from each copy, with no shared ancillary entanglement.
Suppose that transposing a set $S$ of sites in the computational basis gives a stabilizer mixture $\rho^{T_S}$.
Fix a branch of a local stabilizer protocol on $m$ copies and conjugate the local Kraus operators and ancillary states at the sites in $S$.
This is still a stabilizer protocol, since complex conjugation preserves the Clifford group, Pauli measurements and stabilizer states.
Applied to $(\rho^{T_S})^{\otimes m}$, the conjugated branch gives the partial transpose of the original output with the same probability.
Conditioned on this branch, any one-site output is therefore a marginal of a stabilizer mixture or its transpose, both of which are stabilizer mixtures.
For an even number of identically relaxing qubits in a real cat trajectory, transposing half the sites moves the coherence $c$ between two strings of equal population $d$.
Full separability gives $d\ge c$ by positivity of this partial transpose. Subtracting $2c$ times the projector onto their equal positive superposition, a stabilizer state, leaves a nonnegative diagonal operator.
Thus no site can extract a nonstabilizer qubit from any finite number of copies at a fully separable point. Supplementary Note~\ref{si:local} gives the extraction protocol for an entangled pair on the reborn branch.

\subsection{Numerical checks and device estimate}
For small registers, the signed-decomposition robustness was evaluated by minimizing \(\sum_j(u_j+v_j)\) subject to
\(\rho=\sum_j(u_j-v_j)\ket{s_j}\!\bra{s_j}\), \(u_j,v_j\ge0\), over the pure stabilizer states \(\ket{s_j}\)~\cite{Howard2017Application,Heinrich2019Robustness}. The Pauli coordinates of these projectors give a real linear program. For one through four qubits the programs account for all \(6\), \(60\), \(1080\) and \(36{,}720\) pure stabilizer states, respectively. For the Dicke states in Supplementary Fig.~\ref{si:figdicke}, permutation symmetry groups the stabilizer vertices into orbits. 

For the two-qubit device estimate, let \(q_i=\exp(-t/T_{1,i})\),
\(U_i=\vartheta_i(1-q_i)\) and \(B_i=U_i+q_i\). The initial
ground component gives independent excitation probabilities \(U_i\),
and the initial excited component gives probabilities \(B_i\).
The diagonal of the output is the mixture of these two product
distributions with weights \(\alpha^2,\beta^2\).
Its coherence is \(c=\alpha\beta\exp[-t(1/T^*_{2,1}+1/T^*_{2,2})]\).
The first magic loss solves \(c=P_{00}\), the return solves \(c=P_{11}\),
and the separability threshold solves \(c^2=P_{01}P_{10}\).
This calculation assumes the published single-qubit calibrations apply
throughout the idle interval.

\section*{Acknowledgements}
I thank Jens Eisert for many discussions and for the supportive research environment in Berlin, Qi Zhao for hosting me during part of this work and for stimulating discussions, and Zidan Wang for supporting my research at HKIQST. I also thank Salvatore F.\ E.\ Oliviero, Zhenhuan Liu, Lennart Bittel, Ryotaro Suzuki, Yifan Tang and Jue Xu for feedback on the manuscript. This work was supported by the Alexander von Humboldt Foundation.

\makeatletter
\if@filesw\immediate\write\@auxout{\string\citation{arxivBibliography}}\fi
\makeatother
\bibliographystyle{apsrev4-2}
\bibliography{bibliography}

@misc{Wei2024Noise,
  author       = {Fuchuan Wei and Zi-Wen Liu},
  title        = {Noise robustness and threshold of many-body quantum magic},
  year         = {2024},
  howpublished = {Preprint at \url{https://arxiv.org/abs/2410.21215}}
}

@article{Yu2004Finite,
  author       = {Yu, Ting and Eberly, J. H.},
  title        = {Finite-Time Disentanglement Via Spontaneous Emission},
  journal      = {Phys. Rev. Lett.},
  volume       = {93},
  pages        = {140404},
  year         = {2004},
  doi          = {10.1103/PhysRevLett.93.140404}
}

@article{Yu2009Sudden,
  author       = {Ting Yu and J. H. Eberly},
  title        = {Sudden Death of Entanglement},
  journal      = {Science},
  volume       = {323},
  pages        = {598--601},
  year         = {2009},
  doi          = {10.1126/science.1167343}
}

@article{Bravyi2005Universal,
  author       = {Bravyi, Sergey and Kitaev, Alexei},
  title        = {Universal quantum computation with ideal {C}lifford gates and noisy ancillas},
  journal      = {Phys. Rev. A},
  volume       = {71},
  pages        = {022316},
  year         = {2005},
  doi          = {10.1103/PhysRevA.71.022316}
}

@article{Veitch2014The,
  author       = {Veitch, Victor and Hamed Mousavian, S A and Gottesman, Daniel and Emerson, Joseph},
  title        = {The resource theory of stabilizer quantum computation},
  journal      = {New J. Phys.},
  volume       = {16},
  pages        = {013009},
  year         = {2014},
  doi          = {10.1088/1367-2630/16/1/013009}
}

@article{Howard2017Application,
  author       = {Howard, Mark and Campbell, Earl},
  title        = {Application of a Resource Theory for Magic States to Fault-Tolerant Quantum Computing},
  journal      = {Phys. Rev. Lett.},
  volume       = {118},
  pages        = {090501},
  year         = {2017},
  doi          = {10.1103/PhysRevLett.118.090501}
}

@article{Heinrich2019Robustness,
  author       = {Heinrich, Markus and Gross, David},
  title        = {Robustness of magic and symmetries of the stabiliser polytope},
  journal      = {Quantum},
  volume       = {3},
  pages        = {132},
  year         = {2019},
  doi          = {10.22331/q-2019-04-08-132}
}

@article{Trigueros2025Nonstabilizerness,
  author       = {Ballar Trigueros, Fabian and Mar{\'i}n Guzm{\'a}n, Jos{\'e} Antonio},
  title        = {Nonstabilizerness and error resilience in noisy quantum circuits},
  journal      = {Phys. Rev. Lett.},
  volume       = {136},
  pages        = {240602},
  year         = {2026},
  doi          = {10.1103/cbjz-x45n}
}

@misc{Leone2026Unbearable,
  author       = {Lorenzo Leone and Jens Eisert and Salvatore F. E. Oliviero},
  title        = {The unbearable hardness of deciding about magic},
  year         = {2026},
  howpublished = {Preprint at \url{https://arxiv.org/abs/2602.22330}}
}

@article{Aaronson2004Improved,
  author       = {Aaronson, Scott and Gottesman, Daniel},
  title        = {Improved simulation of stabilizer circuits},
  journal      = {Phys. Rev. A},
  volume       = {70},
  pages        = {052328},
  year         = {2004},
  doi          = {10.1103/PhysRevA.70.052328}
}

@article{Wootters1998Entanglement,
  author       = {Wootters, William K.},
  title        = {Entanglement of Formation of an Arbitrary State of Two Qubits},
  journal      = {Phys. Rev. Lett.},
  volume       = {80},
  pages        = {2245--2248},
  year         = {1998},
  doi          = {10.1103/PhysRevLett.80.2245}
}

@article{Almeida2007Environment,
  author       = {M. P. Almeida and F. de Melo and M. Hor-Meyll and A. Salles and S. P. Walborn and P. H. Souto Ribeiro and L. Davidovich},
  title        = {Environment-Induced Sudden Death of Entanglement},
  journal      = {Science},
  volume       = {316},
  pages        = {579--582},
  year         = {2007},
  doi          = {10.1126/science.1139892}
}

@inproceedings{Gottesman1999The,
  author       = {Gottesman, Daniel},
  title        = {The {H}eisenberg Representation of Quantum Computers},
  booktitle    = {Group22: Proceedings of the XXII International Colloquium on Group Theoretical Methods in Physics},
  pages        = {32--43},
  year         = {1999},
  publisher    = {International Press},
  note         = {Preprint at \url{https://arxiv.org/abs/quant-ph/9807006}}
}

@article{Dehaene2003Clifford,
  author       = {Dehaene, Jeroen and De Moor, Bart},
  title        = {Clifford group, stabilizer states, and linear and quadratic operations over {G}{F}(2)},
  journal      = {Phys. Rev. A},
  volume       = {68},
  pages        = {042318},
  year         = {2003},
  doi          = {10.1103/PhysRevA.68.042318}
}

@article{Leone2022Stabilizer,
  author       = {Leone, Lorenzo and Oliviero, Salvatore F. E. and Hamma, Alioscia},
  title        = {Stabilizer {R\'enyi} entropy},
  journal      = {Phys. Rev. Lett.},
  volume       = {128},
  pages        = {050402},
  year         = {2022},
  doi          = {10.1103/PhysRevLett.128.050402}
}

@article{Liu2022Many,
  author       = {Liu, Zi-Wen and Winter, Andreas},
  title        = {Many-Body Quantum Magic},
  journal      = {PRX Quantum},
  volume       = {3},
  pages        = {020333},
  year         = {2022},
  doi          = {10.1103/PRXQuantum.3.020333}
}

@article{Sticlet2025Nonstabilizerness,
  author       = {Sticlet, Doru and D{\'o}ra, Bal{\'a}zs and Szombathy, Dominik and Zar{\'a}nd, Gergely and Moca, C{\u{a}}t{\u{a}}lin Pa{\c{s}}cu},
  title        = {Nonstabilizerness in open {XXZ} spin chains: Universal scaling and dynamics},
  journal      = {Phys. Rev. Res.},
  volume       = {7},
  pages        = {043130},
  year         = {2025},
  doi          = {10.1103/96bk-xf8p}
}

@article{Niroula2024Phase,
  author       = {Niroula, Pradeep and White, Christopher David and Wang, Qingfeng and Johri, Sonika and Zhu, Daiwei and Monroe, Christopher and Noel, Crystal and Gullans, Michael J.},
  title        = {Phase transition in magic with random quantum circuits},
  journal      = {Nat. Phys.},
  volume       = {20},
  pages        = {1786--1792},
  year         = {2024},
  doi          = {10.1038/s41567-024-02637-3}
}

@article{Leone2024Stabilizer,
  author       = {Leone, Lorenzo and Bittel, Lennart},
  title        = {Stabilizer entropies are monotones for magic-state resource theory},
  journal      = {Phys. Rev. A},
  volume       = {110},
  pages        = {L040403},
  year         = {2024},
  doi          = {10.1103/PhysRevA.110.L040403}
}

@article{Warmuz2025Magic,
  author       = {Warmuz, Krzysztof and Dokudowiec, Ernest and Radhakrishnan, Chandrashekar and Byrnes, Tim},
  title        = {Magic Monotone for Faithful Detection of Nonstabilizerness in Mixed States},
  journal      = {Phys. Rev. Lett.},
  volume       = {135},
  pages        = {010203},
  year         = {2025},
  doi          = {10.1103/2s3j-t22p}
}

@article{Reichardt2005Quantum,
  author       = {Reichardt, Ben W.},
  title        = {Quantum universality from magic states distillation applied to {CSS} codes},
  journal      = {Quantum Inf. Process.},
  volume       = {4},
  pages        = {251--264},
  year         = {2005},
  doi          = {10.1007/s11128-005-7654-8}
}

@article{Haug2026Efficient,
  author       = {Haug, Tobias and Tarabunga, Poetri Sonya},
  title        = {Efficient witnessing and testing of magic in mixed quantum states},
  journal      = {npj Quantum Inf.},
  volume       = {12},
  pages        = {40},
  year         = {2026},
  doi          = {10.1038/s41534-026-01189-z}
}

@article{Liu2025A,
  author       = {Zhenhuan Liu and Tobias Haug and Qi Ye and Zi-Wen Liu and Ingo Roth},
  title        = {Triangle {C}riterion: a mixed-state magic criterion with applications in distillation and detection},
  journal      = {PRX Quantum},
  volume       = {7},
  pages        = {033016},
  year         = {2026},
  doi          = {10.1103/rcpf-8nh9}
}

@article{Seddon2021Quantifying,
  author       = {Seddon, James R. and Regula, Bartosz and Pashayan, Hakop and Ouyang, Yingkai and Campbell, Earl T.},
  title        = {Quantifying Quantum Speedups: Improved Classical Simulation From Tighter Magic Monotones},
  journal      = {PRX Quantum},
  volume       = {2},
  pages        = {010345},
  year         = {2021},
  doi          = {10.1103/PRXQuantum.2.010345}
}

@article{Aolita2008Scaling,
  author       = {Aolita, L. and Chaves, R. and Cavalcanti, D. and Ac\'{\i}n, A. and Davidovich, L.},
  title        = {Scaling Laws for the Decay of Multiqubit Entanglement},
  journal      = {Phys. Rev. Lett.},
  volume       = {100},
  pages        = {080501},
  year         = {2008},
  doi          = {10.1103/PhysRevLett.100.080501}
}

@article{Korbany2025Long,
  author       = {Korbany, David Aram and Gullans, Michael J. and Piroli, Lorenzo},
  title        = {Long-Range Nonstabilizerness and Phases of Matter},
  journal      = {Phys. Rev. Lett.},
  volume       = {135},
  pages        = {160404},
  year         = {2025},
  doi          = {10.1103/1hlj-h6t9}
}

@misc{Wei2026Long,
  author       = {Fuchuan Wei and Zi-Wen Liu},
  title        = {Long-range nonstabilizerness and quantum codes, phases, and complexity},
  year         = {2026},
  howpublished = {Preprint at \url{https://arxiv.org/abs/2503.04566v3}}
}

@article{Rafsanjani2012Genuinely,
  author       = {Hashemi Rafsanjani, S. M. and Huber, M. and Broadbent, C. J. and Eberly, J. H.},
  title        = {Genuinely multipartite concurrence of {$N$}-qubit {$X$} matrices},
  journal      = {Phys. Rev. A},
  volume       = {86},
  pages        = {062303},
  year         = {2012},
  doi          = {10.1103/PhysRevA.86.062303}
}

@article{Aolita2015Open,
  author       = {Aolita, Leandro and de Melo, Fernando and Davidovich, Luiz},
  title        = {Open-system dynamics of entanglement: a key issues review},
  journal      = {Rep. Prog. Phys.},
  volume       = {78},
  pages        = {042001},
  year         = {2015},
  doi          = {10.1088/0034-4885/78/4/042001}
}

@article{Seddon2019Quantifying,
  author       = {Seddon, James R. and Campbell, Earl T.},
  title        = {Quantifying magic for multi-qubit operations},
  journal      = {Proc. R. Soc. A},
  volume       = {475},
  pages        = {20190251},
  year         = {2019},
  doi          = {10.1098/rspa.2019.0251}
}

@article{Tarabunga2025Magic,
  author       = {Tarabunga, Poetri Sonya and Tirrito, Emanuele},
  title        = {Magic transition in measurement-only circuits},
  journal      = {npj Quantum Inf.},
  volume       = {11},
  pages        = {166},
  year         = {2025},
  doi          = {10.1038/s41534-025-01104-y}
}

@article{Dur2000Three,
  author       = {D\"ur, W. and Vidal, G. and Cirac, J. I.},
  title        = {Three qubits can be entangled in two inequivalent ways},
  journal      = {Phys. Rev. A},
  volume       = {62},
  pages        = {062314},
  year         = {2000},
  doi          = {10.1103/PhysRevA.62.062314}
}

@article{Zyczkowski2001Dynamics,
  author       = {{\.Z}yczkowski, Karol and Horodecki, Pawe{\l} and Horodecki, Micha{\l} and Horodecki, Ryszard},
  title        = {Dynamics of quantum entanglement},
  journal      = {Phys. Rev. A},
  volume       = {65},
  pages        = {012101},
  year         = {2001},
  doi          = {10.1103/PhysRevA.65.012101}
}

@article{Salles2008Experimental,
  author       = {Salles, A. and de Melo, F. and Almeida, M. P. and Hor-Meyll, M. and Walborn, S. P. and Souto Ribeiro, P. H. and Davidovich, L.},
  title        = {Experimental investigation of the dynamics of entanglement: Sudden death, complementarity, and continuous monitoring of the environment},
  journal      = {Phys. Rev. A},
  volume       = {78},
  pages        = {022322},
  year         = {2008},
  doi          = {10.1103/PhysRevA.78.022322}
}

@article{Pashayan2015Estimating,
  author       = {Pashayan, Hakop and Wallman, Joel J. and Bartlett, Stephen D.},
  title        = {Estimating Outcome Probabilities of Quantum Circuits Using Quasiprobabilities},
  journal      = {Phys. Rev. Lett.},
  volume       = {115},
  pages        = {070501},
  year         = {2015},
  doi          = {10.1103/PhysRevLett.115.070501}
}

@article{Oliviero2022Measuring,
  author       = {Oliviero, Salvatore F. E. and Leone, Lorenzo and Hamma, Alioscia and Lloyd, Seth},
  title        = {Measuring magic on a quantum processor},
  journal      = {npj Quantum Inf.},
  volume       = {8},
  pages        = {148},
  year         = {2022},
  doi          = {10.1038/s41534-022-00666-5}
}

@article{Dur2000Classification,
  author       = {D\"ur, W. and Cirac, J. I.},
  title        = {Classification of multiqubit mixed states: Separability and distillability properties},
  journal      = {Phys. Rev. A},
  volume       = {61},
  pages        = {042314},
  year         = {2000},
  doi          = {10.1103/PhysRevA.61.042314}
}

@article{Qian2025Quantum,
  author       = {Qian, Dongheng and Wang, Jing},
  title        = {Quantum nonlocal nonstabilizerness},
  journal      = {Phys. Rev. A},
  volume       = {111},
  pages        = {052443},
  year         = {2025},
  doi          = {10.1103/PhysRevA.111.052443}
}

@article{Bellomo2007NonMarkovian,
  author       = {Bellomo, B. and Lo Franco, R. and Compagno, G.},
  title        = {Non-{M}arkovian Effects on the Dynamics of Entanglement},
  journal      = {Phys. Rev. Lett.},
  volume       = {99},
  pages        = {160502},
  year         = {2007},
  doi          = {10.1103/PhysRevLett.99.160502}
}

@article{Lopez2008Sudden,
  author       = {L\'opez, C. E. and Romero, G. and Lastra, F. and Solano, E. and Retamal, J. C.},
  title        = {Sudden Birth versus Sudden Death of Entanglement in Multipartite Systems},
  journal      = {Phys. Rev. Lett.},
  volume       = {101},
  pages        = {080503},
  year         = {2008},
  doi          = {10.1103/PhysRevLett.101.080503}
}

@article{Dicke1954Coherence,
  author       = {Dicke, R. H.},
  title        = {Coherence in Spontaneous Radiation Processes},
  journal      = {Phys. Rev.},
  volume       = {93},
  pages        = {99--110},
  year         = {1954},
  doi          = {10.1103/PhysRev.93.99}
}

@article{Scocco2026Rise,
  author       = {Scocco, Annarita and Mok, Wai-Keong and Aolita, Leandro and Collura, Mario and Haug, Tobias},
  title        = {Rise and fall of nonstabilizerness via random measurements},
  journal      = {Phys. Rev. Res.},
  volume       = {8},
  pages        = {013217},
  year         = {2026},
  doi          = {10.1103/31sq-k4m3}
}

@article{Turkeshi2025Magic,
  author       = {Turkeshi, Xhek and Tirrito, Emanuele and Sierant, Piotr},
  title        = {Magic spreading in random quantum circuits},
  journal      = {Nat. Commun.},
  volume       = {16},
  pages        = {2575},
  year         = {2025},
  doi          = {10.1038/s41467-025-57704-x}
}

@article{Junior2026Trading,
  author       = {de Oliveira Junior, A. and Mac{\^e}do, Rafael A. and Czartowski, Jakub and Brask, Jonatan Bohr and Chaves, Rafael},
  title        = {Trading athermality for nonstabilizerness},
  journal      = {Phys. Rev. A},
  volume       = {113},
  pages        = {052421},
  year         = {2026},
  doi          = {10.1103/hwfq-ytkb}
}

@misc{Wei2026ENM,
  author       = {Fuchuan Wei and Ruixia Wang and Yujia Zhang and Huihui Li and Junfeng Li and Fei Yan and Zi-Wen Liu},
  title        = {Entirely nonlocal quantum magic without entanglement},
  year         = {2026},
  howpublished = {Preprint at \url{https://arxiv.org/abs/2608.25950}}
}

@misc{Sierant2026Exact,
  author       = {Piotr Sierant},
  title        = {Exact quantification of nonlocal magic},
  year         = {2026},
  howpublished = {Preprint at \url{https://arxiv.org/abs/2608.28563}}
}

@misc{deOliveira2026Symmetry,
  author       = {de Oliveira Junior, A. and Xuereb, Jake and Mac{\^e}do, Rafael A. and Brask, Jonatan Bohr and Chaves, Rafael},
  title        = {When symmetry suppresses magic},
  year         = {2026},
  howpublished = {Preprint at \url{https://arxiv.org/abs/2609.38319}}
}

@misc{Bozzetto2026Coherent,
  author       = {Bozzetto, Enrico and Hance, Jonte R.},
  title        = {Coherent errors, syndrome extraction and the robustness of magic in stabiliser codes},
  year         = {2026},
  howpublished = {Preprint at \url{https://arxiv.org/abs/2610.38554}}
}

@misc{CaoZhao2026Records,
  author       = {Cao, Chenfeng and Zhao, Qi},
  title        = {Environmental records unlock universal quantum computation from thermal decoherence},
  year         = {2026},
  howpublished = {Preprint at \url{https://arxiv.org/abs/2609.16460}}
}

@article{Garner2026Thermal,
  author       = {Garner, Sean R. and Myers, Nathan M. and Wang, Meng and Stein, Samuel and Liu, Chenxu and Li, Ang},
  title        = {Simulation of thermal-relaxation noise for quantum error correction},
  journal      = {Phys. Rev. Res.},
  volume       = {8},
  pages        = {023134},
  year         = {2026},
  doi          = {10.1103/hgnj-v4j7}
}

@article{Streltsov2011Discord,
  author       = {Streltsov, Alexander and Kampermann, Hermann and Bru{\ss}, Dagmar},
  title        = {Behavior of quantum correlations under local noise},
  journal      = {Phys. Rev. Lett.},
  volume       = {107},
  pages        = {170502},
  year         = {2011},
  doi          = {10.1103/PhysRevLett.107.170502}
}

@article{Ciccarello2012Discord,
  author       = {Ciccarello, Francesco and Giovannetti, Vittorio},
  title        = {Creating quantum correlations through local nonunitary memoryless channels},
  journal      = {Phys. Rev. A},
  volume       = {85},
  pages        = {010102(R)},
  year         = {2012},
  doi          = {10.1103/PhysRevA.85.010102}
}

@article{Lacroix2020Continuous,
  author       = {Lacroix, Nathan and Hellings, Christoph and Andersen, Christian Kraglund and Di Paolo, Agustin and Remm, Ants and Lazar, Stefania and Krinner, Sebastian and Norris, Graham J. and Gabureac, Mihai and Heinsoo, Johannes and Blais, Alexandre and Eichler, Christopher and Wallraff, Andreas},
  title        = {Improving the performance of deep quantum optimization algorithms with continuous gate sets},
  journal      = {PRX Quantum},
  volume       = {1},
  pages        = {020304},
  year         = {2020},
  doi          = {10.1103/PRXQuantum.1.020304}
}

@article{Monz2011Fourteen,
  author       = {Monz, Thomas and Schindler, Philipp and Barreiro, Julio T. and Chwalla, Michael and Nigg, Daniel and Coish, William A. and Harlander, Maximilian and H{\"a}nsel, Wolfgang and Hennrich, Markus and Blatt, Rainer},
  title        = {14-qubit entanglement: Creation and coherence},
  journal      = {Phys. Rev. Lett.},
  volume       = {106},
  pages        = {130506},
  year         = {2011},
  doi          = {10.1103/PhysRevLett.106.130506}
}

@misc{Lyu2026Fate,
  author       = {Lyu, Liuke and S{\o}rensen, Erik S. and Witczak-Krempa, William},
  title        = {Fate of quantum magic},
  year         = {2026},
  howpublished = {Preprint at \url{https://arxiv.org/abs/2610.39059}}
}

@misc{Cao2026RebirthV2,
  author       = {Cao, Chenfeng},
  title        = {Sudden death of entanglement, rebirth of magic},
  year         = {2026},
  howpublished = {Preprint at \url{https://arxiv.org/abs/2605.22603v2}}
}

@inproceedings{Bakshi2024High,
  author       = {Bakshi, Ainesh and Liu, Allen and Moitra, Ankur and Tang, Ewin},
  title        = {High-temperature {Gibbs} states are unentangled and efficiently preparable},
  booktitle    = {2024 IEEE 65th Annual Symposium on Foundations of Computer Science (FOCS)},
  pages        = {1027--1036},
  year         = {2024},
  publisher    = {IEEE},
  doi          = {10.1109/FOCS61266.2024.00068}
}

@article{Ramkumar2026Fermionic,
  author       = {Ramkumar, Akshar and Cai, Yiyi and Tong, Yu and Jiang, Jiaqing},
  title        = {High-temperature fermionic {Gibbs} states are mixtures of {Gaussian} states},
  journal      = {Phys. Rev. Lett.},
  volume       = {137},
  pages        = {020601},
  year         = {2026},
  doi          = {10.1103/qprk-k7wn}
}

@misc{Putterman2026Thermal,
  author       = {Putterman, Harald and Zlokapa, Alexander and Cotler, Jordan},
  title        = {When quantum thermal states look classical},
  year         = {2026},
  howpublished = {Preprint at \url{https://arxiv.org/abs/2607.28536}}
}

@article{Bennett1996Purification,
  author       = {Bennett, Charles H. and Brassard, Gilles and Popescu, Sandu and Schumacher, Benjamin and Smolin, John A. and Wootters, William K.},
  title        = {Purification of Noisy Entanglement and Faithful Teleportation via Noisy Channels},
  journal      = {Phys. Rev. Lett.},
  volume       = {76},
  pages        = {722--725},
  year         = {1996},
  doi          = {10.1103/PhysRevLett.76.722}
}

@article{Ha2019Separability,
  author       = {Ha, Kil-Chan and Han, Kyung Hoon and Kye, Seung-Hyeok},
  title        = {Separability of multi-qubit states in terms of diagonal and anti-diagonal entries},
  journal      = {Quantum Inf. Process.},
  volume       = {18},
  pages        = {34},
  year         = {2019},
  doi          = {10.1007/s11128-018-2145-x}
}

@CONTROL{arxivBibliography, title="0"}

\clearpage
\onecolumngrid
\section*{Supplementary Information}
\setcounter{section}{0}
\setcounter{subsection}{0}
\setcounter{subsubsection}{0}
\setcounter{equation}{0}
\setcounter{figure}{0}
\setcounter{table}{0}
\renewcommand{\thesection}{S\arabic{section}}
\renewcommand{\theequation}{S\arabic{equation}}
\renewcommand{\thefigure}{S\arabic{figure}}
\renewcommand{\thetable}{S\arabic{table}}
\renewcommand{\thesubsection}{\thesection.\arabic{subsection}}
\renewcommand{\thesubsubsection}{\thesubsection.\arabic{subsubsection}}
\makeatletter
\def\p@subsection{}
\def\p@subsubsection{}
\makeatother
\providecommand{\theHsection}{}
\renewcommand{\theHsection}{supp.\arabic{section}}
\providecommand{\theHequation}{}
\renewcommand{\theHequation}{supp.\arabic{equation}}
\providecommand{\theHfigure}{}
\renewcommand{\theHfigure}{supp.\arabic{figure}}
\providecommand{\theHtable}{}
\renewcommand{\theHtable}{supp.\arabic{table}}
\renewcommand{\theHsubsection}{\theHsection.\arabic{subsection}}
\renewcommand{\theHsubsubsection}{\theHsubsection.\arabic{subsubsection}}

\noindent These notes supply derivations and examples for the two theorems and the readout protocols in the article. Note~\ref{si:membership} proves the population test used throughout and establishes its stability at fixed register size. Note~\ref{si:thermal} follows a fixed cat input through the thermal phase diagram of Theorem~\ref{thm:min_rebirth_size}. Note~\ref{si:inputs} treats inputs beyond the cat family, including a three-qubit ground state with a finite interval of zero magic. Note~\ref{si:mirror} determines when the two reservoir and register thresholds coincide. Note~\ref{si:conditional} treats parity readout with a noisy classical record and gives an error budget for the joint decoder. Note~\ref{si:local} determines what separated laboratories can extract from one or more copies.

\section{Stabilizer membership and robustness}
\label{si:membership}
For an $n$-qubit register let $\mathcal S$ be the convex hull of its pure stabilizer states. The robustness $\mathcal R(\rho)$ is the smallest $\sum_j|a_j|$ among signed decompositions $\rho=\sum_j a_j\sigma_j$ with $\sigma_j\in\mathcal S$. Its linear-programming dual maximizes $\Tr(W\rho)$ over Hermitian operators with $|\Tr(W\sigma)|\le1$ for every $\sigma\in\mathcal S$~\cite{Howard2017Application}.
In the computational basis a pure stabilizer state has equal-modulus amplitudes on an affine subspace of bit strings and zero amplitude elsewhere. Relative phases on the support are fourth roots of unity~\cite{Dehaene2003Clifford}. For any pair of different strings $x,y$, triangle inequalities on a convex decomposition therefore give
\begin{equation}
 \nu(\sigma_{xy})\le\min(\sigma_{xx},\sigma_{yy}),\qquad
 \nu(c)=|\Re c|+|\Im c|,\qquad \sigma\in\mathcal S.
 \label{si:pair}
\end{equation}
Violation witnesses magic in any density operator, even if it has many other coherences.

\subsection{One pair of coherent configurations}
Suppose that $c=\langle0^n|\rho|1^n\rangle$ is real and that all other off-diagonal pairs vanish. Write $P_0,P_n$ for the corresponding populations and $m=\min(P_0,P_n)$. If $|c|\le m$, assigning weight $2|c|$ to $\ket{G_s}=(\ket{0^n}+s\ket{1^n})/\sqrt2$, $s=\sgn(c)$, leaves nonnegative residual populations on all basis states. Either sign may be used at $c=0$. Together with Eq.~\eqref{si:pair}, this proves the exact membership test. Outside the polytope, put $d=|c|-m>0$ and $\omega=(\rho+d\ket{G_{-s}}\!\bra{G_{-s}})/(1+d)$. Its coherence has magnitude $(|c|-d/2)/(1+d)$, exactly its smaller endpoint population. Thus $\omega\in\mathcal S$, and the signed decomposition $(1+d)\omega-d\ket{G_{-s}}\!\bra{G_{-s}}$ has cost $1+2d$.

To certify this cost, choose $j=0$ or $1$ such that $\rho_{j^n,j^n}=m$, and put
\begin{equation}
 W=\id+s(\ket{0^n}\!\bra{1^n}+\ket{1^n}\!\bra{0^n})-2\ket{j^n}\!\bra{j^n} .
 \label{si:witness}
\end{equation}
If a pure stabilizer state has only one of the two endpoints in its support, or neither, then $|\langle W\rangle|\le1$. If both are present with population $b\le1/2$ each, then $1-4b\le\langle W\rangle\le1$. Hence $W$ is dual feasible. Its expectation in $\rho$ is $1+2(|c|-m)$, which together with the identity witness proves
\begin{equation}
 \mathcal R(\rho)=1+2[|c|-m]_+,\qquad [x]_+=\max\{x,0\} .
 \label{si:rom}
\end{equation}
The two faces in Fig.~\ref{si:figwedge} give different ways to lose or regain magic. Here and below $\gamma_-$ denotes the first magic death, $\gamma_+$ its return, and $\gamma_e$ the loss of entanglement. A superscript $(n)$, when used in a figure, specifies the register size.

\begin{figure}[tbp]
\centering\includegraphics[width=.94\linewidth]{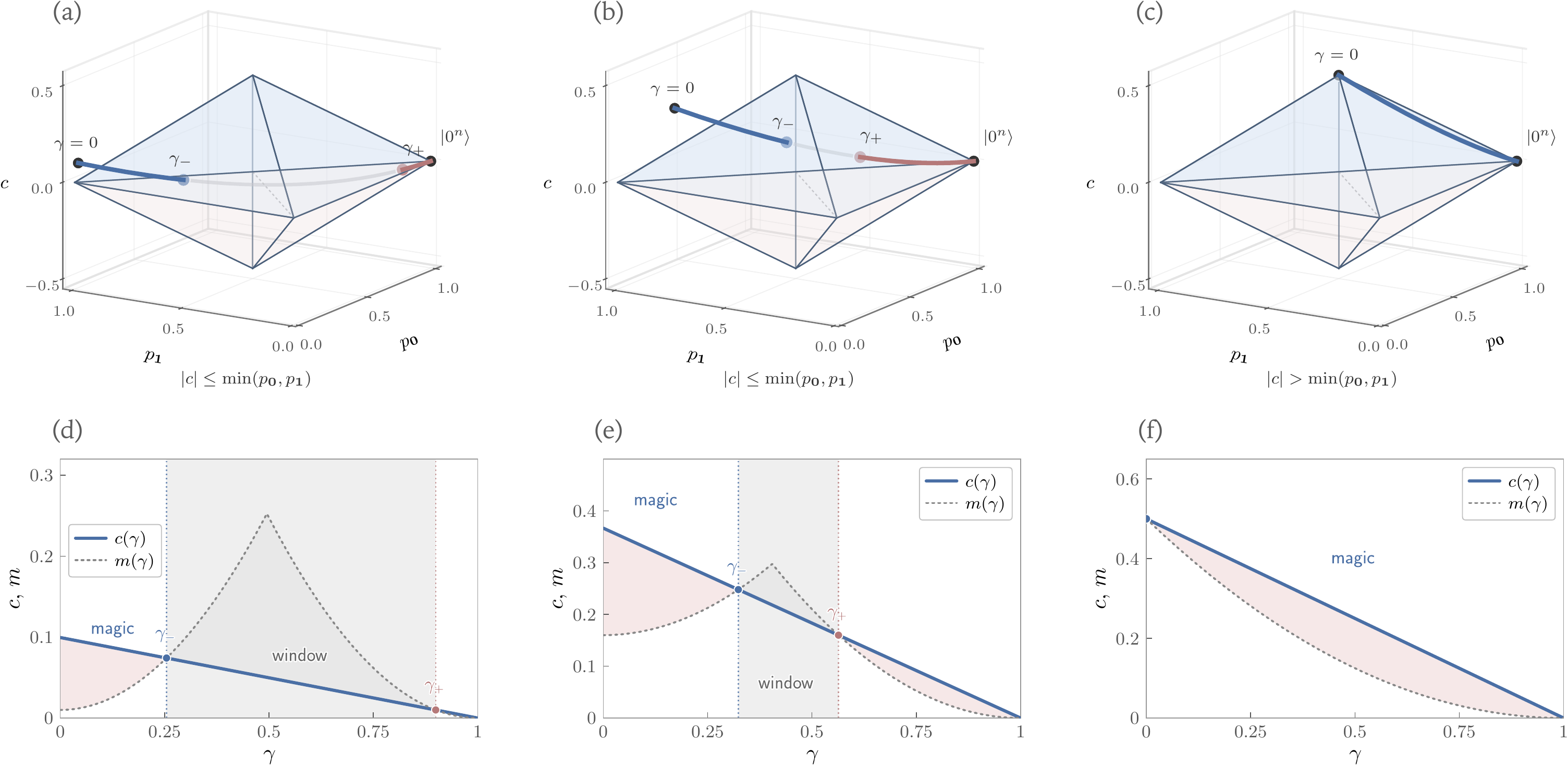}
\caption{\textbf{Population test along two-qubit cat trajectories.} The upper panels plot the stabilizer region in the three coordinates $P_{00},P_{11},c$ (endpoint populations labeled $p_0,p_1$ in the plots) and follow a trajectory through it. The lower panels compare the coherence with the smaller population. The columns use $\alpha=0.1,0.4,1/\sqrt2$, respectively. The first two inputs pass through the interior between two crossings. The balanced input is a stabilizer state at $\gamma=0$ and has magic for $0<\gamma<1$.}
\label{si:figwedge}
\end{figure}

\subsection{Several complementary pairs}
A generalized $X$ state has a computational-basis diagonal $p_x$ and coherences $c_x$ only between $x$ and its bitwise complement $\bar x$. The $n$-qubit formula below also appears in de Oliveira Junior \emph{et al.}~\cite{deOliveira2026Symmetry}; adding the conditional robustnesses is the repetition-code case of a syndrome identity for general stabilizer codes~\cite{Bozzetto2026Coherent}. For completeness, take the complement representative $x=(0,s)$ for each $(n-1)$-bit word $s$. Applying a CNOT from the first qubit to every other qubit gives the direct sum
\begin{equation}
 U\rho U^\dagger=\sum_s \tau_s\otimes\ket{s}\!\bra{s},\qquad
 \tau_s=\begin{pmatrix}p_x&c_x\\c_x^*&p_{\bar x}\end{pmatrix},\qquad w_s=\Tr\tau_s.
\end{equation}
For $w_s>0$ set $\sigma_s=\tau_s/w_s$. Tensoring each optimum one-qubit decomposition with its computational flag gives $\mathcal R(\rho)\le\sum_s w_s\mathcal R(\sigma_s)$. Conversely, start with a signed decomposition $\sum_j a_j\sigma_j'$ of $\rho$ and measure the flags after $U$. The probability $p(s|j)$ for a given component multiplies a normalized stabilizer mixture, so $w_s\mathcal R(\sigma_s)\le\sum_j|a_j|p(s|j)$. Summation over $s$ proves the reverse bound. The one-qubit formula $\mathcal R(\sigma)=\max(1,|m_x|+|m_y|+|m_z|)$ in Bloch coordinates then gives
\begin{equation}
 \mathcal R(\rho)=\sum_{s:w_s>0}w_s\mathcal R(\sigma_s)
 =1+2\sum_{x:x_1=0}[\nu(c_x)-\min(p_x,p_{\bar x})]_+ .
 \label{si:xrom}
\end{equation}
A jump on any site annihilates $\ket{x}\!\bra{\bar x}$, so independent amplitude damping with probabilities $\gamma_i$ multiplies every $c_x$ by $\sqrt{\prod_i(1-\gamma_i)}$ and preserves this form.

\subsection{Persistence of magic loss and return under input perturbations}
For the real cat $\alpha\ket{0^n}+\beta\ket{1^n}$ with $n\ge2$, $0<r=\alpha/\beta<1$, choose three damping probabilities $0<\gamma_a<\gamma_-<\gamma_m<\gamma_+<\gamma_b<1$. At $\gamma_a$ and $\gamma_b$ a witness of the form in Eq.~\eqref{si:witness} exceeds one by a positive margin. At $\gamma_m$ the endpoint populations both strictly exceed $c$, and every other computational-basis population is positive. Subtracting $2c\ket{G_+}\!\bra{G_+}$ leaves a strictly positive diagonal operator. It therefore contains a positive multiple of $\id/2^n$ with a stabilizer remainder. The maximally mixed state lies in the interior of $\mathcal S$: it is a strictly positive uniform combination of the product Pauli eigenstates, whose projectors span the Hermitian operators. Thus the output at $\gamma_m$ is interior. By continuity there is an open neighborhood of the input, relative to all density operators, for which the two witness violations and the intervening interior point persist. Continuity in damping gives an open interval of zero magic for each such input. Its width and the size of the neighborhood depend on $n$.

\section{Thermal thresholds for a fixed cat state}
\label{si:thermal}
Theorem~\ref{thm:min_rebirth_size} minimizes the register size over inputs. Here we fix $r=\alpha/\beta\in(0,1)$ and $n$, and vary the idle time and the equilibrium excited-state population $\vartheta$. For $q=e^{-t/T_1}$ and $a=T_1/T_2\ge1/2$, denote the probabilities of an excitation given initial states $\ket0$ and $\ket1$ by $U=\vartheta(1-q)$ and $B=\vartheta+(1-\vartheta)q$, respectively. Write $A=1-U$ and $D=1-B$. The endpoint populations and real coherence are
\begin{equation}
 P_0=\alpha^2A^n+\beta^2D^n,\qquad P_n=\alpha^2U^n+\beta^2B^n,\qquad c=\alpha\beta q^{an}.
 \label{si:thermalend}
\end{equation}
The initial ratio $c/P_0=1/r$ exceeds one. For $0<q<1$ its logarithmic derivative obeys
\begin{equation}
 \frac{1}{n}\frac{d}{dq}\log\frac{c}{P_0}
 \ge \frac a q-\frac{\vartheta}{A}=\frac{aA-\vartheta q}{qA}.
\end{equation}
Indeed $dP_0/dq\le n\alpha^2\vartheta A^{n-1}$ and $P_0\ge\alpha^2A^n$. The numerator is nonnegative because $0\le\vartheta\le1/2\le a$ and $A=1-\vartheta+\vartheta q\ge q$. For $0<q<1$ the term from $D^n$ makes the derivative strictly positive. As $q$ decreases to zero, $c/P_0$ decreases to zero, so the trajectory crosses this stabilizer face exactly once.

At positive temperature $c-P_n$ is negative at both $q=0$ and $q=1$. For $a\ge1$, $q^a\le q\le B$ rules out a return. For $1/2\le a<1$, put $\xi=-\log q$. Then
\begin{equation}
 \frac{P_n}{c}=r\vartheta^n\left(e^{a\xi}-e^{(a-1)\xi}\right)^n
 +r^{-1}\left[\vartheta e^{a\xi}+(1-\vartheta)e^{(a-1)\xi}\right]^n .
 \label{si:convex}
\end{equation}
The first base is nonnegative and convex on $\xi\ge0$, since $a^2e^{a\xi}\ge(1-a)^2e^{(a-1)\xi}$. The second base is positive and strictly convex. Raising to the positive integer power $n$ preserves these inequalities, so $P_n/c$ is strictly convex. Its sublevel set $P_n/c<1$, if nonempty, is one interval. It begins after the initial magic has vanished: by positivity $c>P_n$ implies $P_0>c$, and $P_0=P_n$ at an earlier instant separates the two violations. At that intervening instant $P_0=P_n>c$ because the endpoint block is strictly positive definite for $0<q<1$, as shown in Methods. At zero temperature $P_n/c=r^{-1}q^{(1-a)n}$ instead, so for $a<1$ the branch extends to all later finite times.

For each $0<q<1$, $P_n$ increases strictly with $\vartheta$ while $c$ is independent of $\vartheta$. The critical thermal population $\vartheta_c(a,r,n)$ is the supremum of the equilibrium populations for which $c>P_n$ at some time. For $1/2\le a<1$, the maximum $M(\vartheta)=\max_{0\le q\le1}(c-P_n)$ is continuous, with $M(0)>0$ and $M(1-a)<0$. At any zero of $M$ its maximizing $q$ is interior, since both endpoint values are negative at positive temperature. Strict decrease of $c-P_n$ with $\vartheta$ at every interior $q$ then makes this zero unique. At $\vartheta_c$ the minimum of Eq.~\eqref{si:convex} equals one; the strict inequality needed for rebirth holds only below $\vartheta_c$. The solid curves for fixed inputs in Fig.~\ref{fig:channels} approach the dash-dotted boundary $\vartheta=1-a$ as $n$ increases at fixed $r$. This is the boundary for complete stabilizer preservation of the channel~\cite{CaoZhao2026Records}. In fact $q^a\le q^{1-\vartheta}\le\vartheta+(1-\vartheta)q=B$ when $\vartheta\ge1-a$. Conversely, at $q=e^{-x/n}$ and fixed total exposure $x=nt/T_1>0$, $c-P_n$ converges to $\alpha\beta e^{-ax}-\beta^2e^{-(1-\vartheta)x}$. This is positive for $x>\log(1/r)/(1-\vartheta-a)$ when $\vartheta<1-a$. Thus $\vartheta_c\to1-a$ at fixed $r$ for $1/2\le a<1$. The unital case $\vartheta=1/2$ lies on or above this boundary for all physical $a$. Its channel is completely stabilizer-preserving, so a loss of magic is permanent at every register size.

\section{Inputs beyond the cat family}
\label{si:inputs}
Theorem~\ref{thm:late_time_magic} relates late-time magic to the coherences present in the input. An earlier interval of zero magic requires a separate calculation. The three-qubit ground state in the article belongs to an exactly solvable family built from three equal-weight strings and the vacuum. We also determine when damping generates magic from stabilizer inputs and when Dicke states retain it. All damping in this note is at zero temperature without additional dephasing.

\subsection{Two computational configurations}

Consider $\alpha\ket{x}+\beta\ket{y}$ with $x\ne y$, real $\alpha,\beta>0$, and $\alpha^2+\beta^2=1$, under identical pure amplitude damping~\cite{Cao2026RebirthV2}.
Sites on which the two strings agree factor out as a diagonal stabilizer mixture. Appending this mixture and tracing it out are both stabilizer operations, so the remaining sites have the same robustness as the full register.
For comparable excitation sets, label the strings so that the excitations in $x$ form a subset of those in $y$. The remaining strings are $0^d$ and $1^d$, where $d$ is their Hamming distance. The cat solution therefore gives death and rebirth precisely for $\alpha<\beta$, including $d=1$.

If the excitation sets are incomparable, the remaining strings $u$ and $v$ have disjoint, nonempty excitation sets. Neither can decay into the other. For $q>0$ their populations are $p_u=\alpha^2q^{|u|}$ and $p_v=\beta^2q^{|v|}$, and the coherence has modulus $\sqrt{p_up_v}$. The population bound in Eq.~\eqref{si:pair} requires $p_u=p_v$ for stabilizer membership. This condition also suffices. At equality the no-jump branch is the equal superposition of the two strings, a stabilizer state, and each nonzero jump branch is a computational-basis state. The populations can coincide at most once unless $|u|=|v|$ and $\alpha=\beta$, in which case magic is absent throughout. Thus incomparable strings cannot give a finite interval of magic death followed by a return.

\subsection{Rebirth on a punctured binary plane}
As in the earlier preprint~\cite{Cao2026RebirthV2}, let $L=\{0^n,x,y,x\oplus y\}$ be a two-dimensional binary linear subspace whose three nonzero strings each have $k$ excitations. Define $\ket D=\sum_{z\in L\setminus\{0^n\}}\ket z/\sqrt3$ and prepare $u\ket{0^n}+v\ket D$, with $u\ge0$, $v>0$ and $u^2+v^2=1$. Each active site is excited in two of the three strings. Equal weights require the coordinate pairs $(x_i,y_i)=(1,0),(0,1),(1,1)$ to occur $k/2$ times each. Thus $k$ is even, and up to qubit permutations and sites fixed in the vacuum the support repeats each coordinate of the three-qubit pattern $011,101,110$ exactly $k/2$ times.

Under identical amplitude damping with $q=1-\gamma$, the output is
\begin{align}
 \rho={}&p\ket{0^n}\!\bra{0^n}+d\ket D\!\bra D+c_D(\ket{0^n}\!\bra D+\ket D\!\bra{0^n})+\Omega,\nonumber\\
 &p=u^2+v^2\gamma^k,\qquad d=v^2 q^k,\qquad c_D=uvq^{k/2}.
 \label{si:nonx}
\end{align}
The no-jump branch is $u\ket{0^n}+vq^{k/2}\ket D$. A nonempty set of jumps selects at most two of the three excited strings, each with the same surviving amplitude. If fewer than $k$ excitations are lost, each nonzero branch is a basis state or an equal superposition of two strings, hence a stabilizer state. These branches give $\Omega$, supported on weights $1,\ldots,k-1$. Losing all $k$ excitations contributes only to the vacuum.

For $q>0$ the exact membership condition is
\begin{equation}
 \rho\in\mathcal S\quad\Longleftrightarrow\quad p\ge d/3\quad\text{and}\quad c_D\le d/\sqrt3 .
 \label{si:nonxmem}
\end{equation}
For sufficiency set $\ket{S_\pm}=(\pm\ket{0^n}+\sqrt3\ket D)/2$. $\ket{S_+}$ is the uniform stabilizer state on $L$. A quadratic Clifford phase on its two independent binary coordinates reverses only the sign of the vacuum amplitude to give $\ket{S_-}$. The decomposition
\begin{equation}
 \rho=\sum_{\epsilon=\pm1}\left(\frac{2d}{3}+\epsilon\frac{2c_D}{\sqrt3}\right)\ket{S_\epsilon}\!\bra{S_\epsilon}
 +(p-d/3)\ket{0^n}\!\bra{0^n}+\Omega
\end{equation}
has nonnegative weights under the stated inequalities. For necessity, the entire weight-$k$ block is $d\ket D\!\bra D$, of rank one. Every stabilizer component with a nonzero projection onto this block must project onto the same ray. Its affine support therefore contains all three nonzero strings, and their bitwise sum gives $0^n$. Equal-modulus amplitudes imply that a component with weighted mass $3z$ in the block contributes $z$ to the vacuum population and at most $\sqrt3 z$ to the magnitude of the coherence with $\ket D$. Summing over the components gives both inequalities.

For $0<b=u/v<1/\sqrt3$, the state is a stabilizer mixture exactly on $[\gamma_-,\gamma_+]$ and at $\gamma=1$, with
\begin{equation}
 b^2+\gamma_-^k=(1-\gamma_-)^k/3,\qquad \gamma_+=1-(\sqrt3\,b)^{2/k}.
 \label{si:planethresholds}
\end{equation}
Indeed $p-d/3$ increases strictly from a negative value and is positive at $\gamma_+$, where $d/3=u^2$. The three-qubit example in the article has $k=2$, for which $\gamma_-=(\sqrt{3-6b^2}-1)/2$ and $\gamma_+=1-\sqrt3\,b$. Taking $u=1/\sqrt{13}$ and $v=2\sqrt{3/13}$ gives $b=1/(2\sqrt3)$, $\gamma_-=(\sqrt{5/2}-1)/2$ and $\gamma_+=1/2$.

This input is the unique ground state of a two-local Hamiltonian. Write $\hat N=\sum_j(\id-Z_j)/2$, $\sigma^+=\ket1\!\bra0$ and $\sigma^-=\ket0\!\bra1$. Choosing the coefficient of $(\hat N-2)^2$ as the unit of energy and setting $g=8/11$ gives
\begin{equation}
 H=(\hat N-2)^2 -g\sum_{i<j}(\sigma_i^+\sigma_j^+ +\sigma_i^-\sigma_j^-).
\end{equation}
Every term acts on at most two qubits after expanding the square. In the even symmetric basis $(\ket{000},\ket D)$, $H$ has diagonal $(4,0)$ and off-diagonal $-\sqrt3g$. Its lowest eigenvector has amplitude ratio $\sqrt3g/(2+\sqrt{4+3g^2})=1/(2\sqrt3)$ and eigenvalue $-4/11$. The lowest symmetric odd eigenvalue is $1-\sqrt3g>-4/11$; the nonsymmetric even and odd subspaces have energies $0$ and $1$, respectively. This accounts for all eight eigenstates and establishes uniqueness.

\subsection{Which pure stabilizer inputs remain stabilizer?}
For independent positive rates $\kappa_i$, let $w(x)=\sum_{i:x_i=1}\kappa_i$ and $K=\sum_i\kappa_i\ket1_i\!\bra1$. A pure stabilizer input $\ket\phi$ with affine support $A$ has $\operatorname{Var}_\phi(K)=0$ exactly when $w$ is constant on $A$. A jump set $J$ selects those strings with a one at every site in $J$ and maps $x$ to $x\setminus J$ by resetting these bits. The strings in a nonempty branch form an affine slice, and the surviving amplitudes acquire the factor
\begin{equation}
 e^{-tw(x)/2}\prod_{j\in J}\sqrt{e^{\kappa_jt}-1}.
\end{equation}
If $w$ is constant this factor is common to the entire branch, which is just a computational-basis postselection and a set of Pauli flips on the original stabilizer state. Summing the branches gives a stabilizer mixture at every time. Conversely choose a maximal-$w$ word $y\in A$ and $x\in A$ with $w(x)<w(y)$. No nonempty jump can feed $y$, since it would require a word of still larger $w$. Therefore
\begin{equation}
 \frac{|\rho_{xy}(t)|}{\rho_{yy}(t)}=e^{[w(y)-w(x)]t/2}>1\quad(0<t<\infty).
\end{equation}
Equation~\eqref{si:pair} proves magic at every finite positive time. Away from the finite union of rate hyperplanes $w(x)=w(y)$ for distinct words, only single-word supports have zero variance.

With equal damping rates, only $8$ of the $60$ pure two-qubit stabilizer inputs remain stabilizer mixtures throughout relaxation. These are the four basis states and $(\ket{01}+\omega\ket{10})/\sqrt2$ for $\omega\in\{1,-1,i,-i\}$. For three qubits there are $32$ out of $1080$. Besides the eight basis states there are six pairs of equal-weight strings, each with four possible relative phases.

The fraction among all pure $n$-qubit stabilizer states is at most $2^{-n^2/8+O(n)}$. To see this, write a constant-weight affine support as $a+C$, where $C$ is a binary linear space of dimension $m$. Each vector in $C$ flips as many ones of $a$ to zeros as zeros to ones. Its restriction to the one positions of $a$ therefore determines it uniquely, as does its restriction to the zero positions. Hence $m\le\min(|a|,n-|a|)\le n/2$. There are at most $2^{m(n-m)+O(n)}$ affine supports of dimension $m$, from counting binary bases and their translates. On each support the stabilizer phases allow $4^m$ linear and $2^{\binom m2}$ quadratic choices up to a global phase~\cite{Dehaene2003Clifford}. Summing over $m\le n/2$ gives at most $2^{3n^2/8+O(n)}$ states out of $2^n\prod_{j=1}^n(2^j+1)$ in total. For $n=4,5,6$, enumerating the cosets of each binary linear subspace and counting each constant-weight support with its $4^m2^{\binom m2}$ phase choices gives $220$, $1\,432$ and $17\,624$ such states, respectively.

\subsection{Fixed excitation number and Haar-random inputs}
Consider equal damping rates. For an initial real single-excitation superposition $\sum_i b_i\ket{e_i}$ with $b_i\ge0$, the block $M$ on the one-excitation sector is $M_{ij}=qb_ib_j$, and the vacuum has weight $\gamma$. Any pure stabilizer term in this support contains at most two words: three single-excitation words would force a weight-three word by affine closure, and the vacuum with two single-excitation words would force a weight-two word. A real symmetric one-excitation block thus obeys $M_{ii}\ge\sum_{j\ne i}|M_{ij}|$ for every $i$ whenever the state is a stabilizer mixture. Conversely these inequalities give a decomposition using pair superpositions $(\ket{e_i}\pm\ket{e_j})/\sqrt2$ with weights $2|M_{ij}|$ and nonnegative residual diagonal populations. For $q>0$ the inequalities reduce to $2b_i\ge\sum_j b_j$ for each nonzero $b_i$. They hold for a single word or for an equal superposition of two words. Every other input of this form has magic at all finite times.

\begin{figure}[tbp]
\centering\includegraphics[width=.54\linewidth]{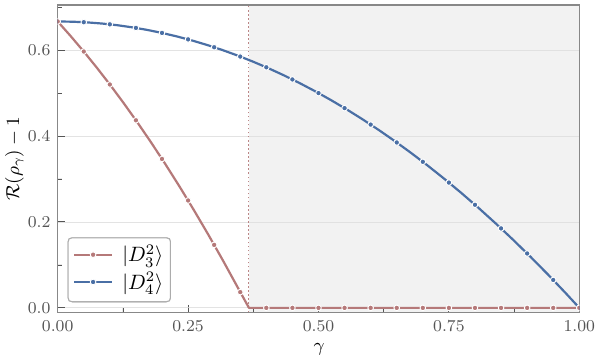}
\caption{\textbf{Magic under relaxation of two Dicke states.} The robustness excess of $\ket{D_3^2}$ is $2[q^2/3-\gamma^2]_+$ and vanishes at $\gamma_*$. For $\ket{D_4^2}$ the blue curve gives the convexity bound $2(1-\gamma^2)/3$. Symbols mark selected damping probabilities. Independent linear programs agree with both curves at all 81 equally spaced values tested on $[0,1]$.}
\label{si:figdicke}
\end{figure}

For a Dicke state $\ket{D_n^k}$~\cite{Dicke1954Coherence}, the equal superposition of all words with $k$ excitations, postselecting $k-1$ specified qubits on $\ket1$ at the output gives the damped $\ket{D_{n-k+1}^{1}}$ on the remaining qubits. Indeed no measured excitation can come from a jump, and the common no-jump factor cancels on normalization. The event has probability $q^{k-1}(n-k+1)/\binom nk>0$ for $q>0$. For $1\le k\le n-2$ the remaining state has at least three terms in its single-excitation amplitude, so it has magic by the preceding test. Stabilizer postselection cannot create this magic. In contrast, the state $\ket D=\ket{D_3^2}$ loses magic at $\gamma_*=(\sqrt3-1)/2$ and keeps zero magic thereafter, by Eq.~\eqref{si:nonxmem} at $u=0$. Both cases have $f=0$ in Theorem~\ref{thm:late_time_magic}, precisely the equality case for pure amplitude damping. Figure~\ref{si:figdicke} illustrates their different behavior.

For $\ket{D_4^2}$, the zero-, one- and two-jump outcomes have total weights $q^2$, $2q\gamma$ and $\gamma^2$. The first leaves the original state, a jump at a specified site leaves $\ket0$ there and the three-qubit $W$ state~\cite{Dur2000Three} on the other sites, and two jumps leave the vacuum. Both $\ket{D_4^2}$ and $\ket{W_3}$ reduce by a Clifford change of basis to the uniform superposition of three two-bit strings times a stabilizer ancilla. Indeed, their supports are an affine space with one of four equal cosets removed. Completing the three-string superposition with the missing string, with either sign and all amplitudes of magnitude $1/2$, gives two stabilizer states. Taking each projector with weight $2/3$ and subtracting the missing string's projector with weight $1/3$ gives the original three-string projector. Both branch costs are therefore at most $5/3$. Convexity then gives $\mathcal R-1\le 2(1-\gamma^2)/3$.

For the three-qubit state at $u=0$ the decomposition following Eq.~\eqref{si:nonxmem} has signed cost $1+2[q^2/3-\gamma^2]_+$. A matching witness on the branch with nonzero excess is $W=\id-2\ket{000}\!\bra{000}-(4/3)\Pi_2+2\ket D\!\bra D$, where $\Pi_2$ projects onto weight two. For a stabilizer state with affine support $A$ of size $N$, write $k=|A\cap\{011,101,110\}|$ and $\delta=\mathbf1_{000\in A}$. Then $\langle W\rangle\le 1+[-2\delta-4k/3+2k^2/3]/N\le1$, since $k=3$ forces $\delta=1$. Also $2\delta+4k/3\le2N$ because $\delta+k\le N$, giving $\langle W\rangle\ge-1$. Its expectation on this trajectory is $1+2q^2/3-2\gamma^2$, as claimed by the analytical curve.

For a Haar-random pure state $\sum_x a_x\ket x$ we can also bound the probability that magic persists for the whole open trajectory. Let $y=1^n$. Only the no-jump operator can contribute to any output element with $y$ on one side, and consequently $\rho_{yy}=|a_y|^2 q^n$ and $\rho_{xy}=a_xa_y^*q^{(|x|+n)/2}$. If any $|a_x|>|a_y|>0$, then $|\rho_{xy}|>\rho_{yy}$ for all $0<q\le1$. The squared amplitudes are exchangeable with a unique maximum almost surely, so this sufficient event has probability $1-2^{-n}$. Thus magic persists to every finite damping time with probability at least $1-2^{-n}$. This bound leaves open the behavior of inputs with largest amplitude at $1^n$.

\subsection{The four-qubit state with one missing excitation}
Under identical pure amplitude damping, the four-qubit anti-$W$ state $\ket{D_4^3}$ loses magic at $\gamma=1/2$ and remains a stabilizer mixture for every larger $\gamma$.
Write $q=1-\gamma$ and let $P(v)=\ket v\!\bra v$.
Here $\ket{D_L^j}$ denotes the Dicke state with $j$ excitations on the site set $L$ and all other sites in $\ket0$.
Grouping the jumps by their number gives the even- and odd-parity blocks
\[
 \rho_{\rm e}=\gamma^3P(0^4)+\frac{3\gamma q^2}{4}\sum_{|L|=3}P(D_L^2),
 \qquad
 \rho_{\rm o}=q^3P(D_4^3)+\frac{\gamma^2q}{2}\sum_{|L|=2}P(D_L^1).
\]

To test the even block, group the six weight-two strings into three complementary pairs.
For amplitudes $z_x$ on these six strings and the vacuum, let $s_j=z_x+z_{\bar x}$ for pair $j$.
Consider the Hermitian form
\[
 W(z)=3|z_{0^4}|^2+2\sum_{j=1}^3|s_j|^2-\left|\sum_{j=1}^3s_j\right|^2 .
\]
It has nonnegative expectation on every stabilizer mixture supported on these seven strings. Indeed, every pure component of such a mixture has support within the same seven strings.
If its support meets at most two complementary pairs, the claim follows from Cauchy--Schwarz.
Otherwise, choose one word from each pair. Viewed as edges on four vertices, the three words form a triangle or a star since they intersect pairwise. Their XOR is therefore either $0000$ or $1111$. Affine closure forces the XOR into the support, so it must be $0000$. The support now consists of precisely these four words, since its size is a power of two. The four amplitudes have a common modulus $h$, which gives $W(z)\ge 3h^2+6h^2-9h^2=0$.
The vacuum gives expectation $3$ and each $D_L^2$ for $|L|=3$ gives $-1$. The expectation on $\rho_{\rm e}$ is thus $3\gamma(\gamma^2-q^2)$. Even parity can be postselected by a Pauli measurement, so a negative value certifies magic for $0<\gamma<1/2$. At $\gamma=0$ the pure input has nonaffine support.

For $\gamma\ge1/2$ both blocks have positive stabilizer decompositions. On each triple $L$, the two states
$\ket{\eta_{L,\pm}}=(\ket{0^4}\pm\sqrt3\ket{D_L^2})/2$ are stabilizers because their support is a binary affine plane and their signs are quadratic phases on its two binary coordinates. The identity
$P(D_L^2)+P(0^4)/3=(2/3)[P(\eta_{L,+})+P(\eta_{L,-})]$
uses $\gamma q^2$ of the available vacuum weight in $\rho_{\rm e}$, leaving $\gamma(\gamma^2-q^2)\ge0$.
For odd parity, $\ket{\chi_\pm}=(\ket{D_4^3}\pm\ket{D_4^1})/\sqrt2$ are also stabilizers. Indeed, $\ket{\chi_+}$ is uniform on the odd-parity hyperplane and applying the phase gate to every qubit gives $\ket{\chi_-}$ up to an overall phase.
Let $\Pi_1$ project onto weight one. Since $\tfrac12\sum_{|L|=2}P(D_L^1)=P(D_4^1)+\Pi_1/2$, the odd block is
\[
 \rho_{\rm o}=q^3\big[P(\chi_+)+P(\chi_-)+\Pi_1/2\big]
       +\frac{q(\gamma^2-q^2)}2\sum_{|L|=2}P(D_L^1).
\]
All coefficients are nonnegative, and every two-site $D_L^1$ is a Bell stabilizer state.

For $n>4$, measuring $n-4$ specified output qubits of the damped $\ket{D_n^{n-1}}$ in the computational basis gives all ones with probability $(4/n)q^{n-4}$. On this outcome the remaining four qubits are precisely the damped $\ket{D_4^3}$, since the unique missing excitation must lie among the four unmeasured sites, and all measured excitations have survived. Thus every $\ket{D_n^{n-1}}$ with $n\ge4$ retains magic throughout $0\le\gamma<1/2$. For $n>4$ this gives a lower bound on the first possible loss of magic.

\section{The reservoir state and the reflection of the two thresholds}
\label{si:mirror}
For a vacuum amplitude-damping bath the local isometry maps $\ket0$ to $\ket{0_S0_E}$ and $\ket1$ to $\sqrt{1-\gamma}\ket{1_S0_E}+\sqrt\gamma\ket{0_S1_E}$. Interchanging $S$ and $E$ thus replaces $\gamma$ by $1-\gamma$. For any input density operator, even with correlations between sites, the reduced state of the reservoirs equals the system output for the complementary set of damping probabilities:
\begin{equation}
 \rho_E(\gamma_1,\ldots,\gamma_n)=\rho_S(1-\gamma_1,\ldots,1-\gamma_n).
 \label{si:complement}
\end{equation}
The computational and reservoir bases are identified through the vacuum and one-excitation levels. For a pure register input the joint output is pure, hence the two outputs in Eq.~\eqref{si:complement} also imply identical nonzero spectra of $\rho_S(\boldsymbol\gamma)$ and $\rho_S(\mathbf1-\boldsymbol\gamma)$. A relation between their stabilizer and entanglement thresholds requires information about matrix elements as well as spectra.

\subsection{Independent damping with unequal rates}
For the real cat, set $G=\prod_i\gamma_i$ and $Q=\prod_i(1-\gamma_i)$. The register has endpoint entries $P_0=\alpha^2+\beta^2 G$, $P_n=\beta^2 Q$, $c=\alpha\beta\sqrt Q$. When all $0<\gamma_i<1$, the boundary $c=P_n$ occurs at $Q=r^2$, at which $P_0-c=\beta^2G>0$. Across a bipartition $A|B$ the partial transpose moves $c$ into the pair $1^A0^B$, $0^A1^B$. The product of their populations, $\beta^4 GQ$, is independent of the cut. The determinant is $\beta^2Q(\beta^2G-\alpha^2)$, so every cut has a negative partial transpose precisely when $G<r^2$. Full separability holds for $G\ge r^2$. To see this, average the product vectors $\bigotimes_i(\sqrt{\gamma_i}\ket0+e^{i\phi_i}\sqrt{1-\gamma_i}\ket1)$ over $n-1$ uniform independent phases with $\phi_n=-\sum_{i<n}\phi_i$. The average has the diagonal obtained from $\ket{1^n}$ and a single complementary coherence $\sqrt{GQ}$. Mixing it with its dephased version in proportions $r/\sqrt G$ and $1-r/\sqrt G$ gives coherence $r\sqrt Q$ and the same diagonal. Multiplication by $\beta^2$ and addition of $\alpha^2\ket{0^n}\!\bra{0^n}$ yields exactly the register. If any site is completely damped, the output is diagonal and the endpoints need no separate witness.
At a fully separable point on the reborn branch, $r\le\sqrt G$ and $c>P_n$, so
\[
 \mathcal R-1=2\beta^2\sqrt Q\,(r-\sqrt Q)
 \le 2\sqrt{GQ}
 =2\prod_i\sqrt{\gamma_i(1-\gamma_i)}
 \le 2^{1-n}.
\]

\begin{figure}[!b]
\centering\includegraphics[width=.96\linewidth]{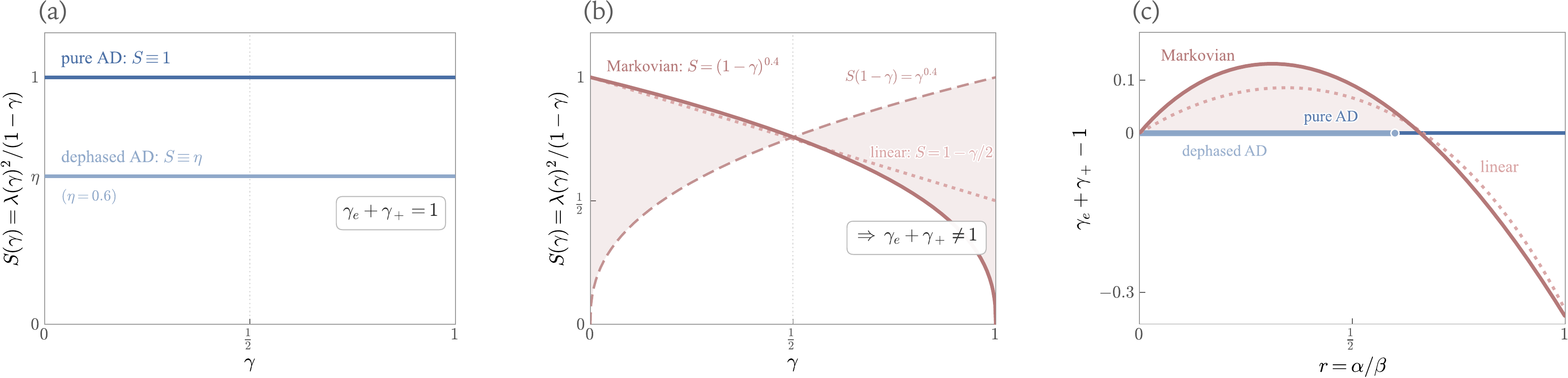}
\caption{\textbf{How transverse decay changes the threshold sum.} (a) Pure amplitude damping (AD) gives $S=1$. The fixed dephasing layer has a constant profile $S(\gamma)=\eta$, with $\eta=0.6$. (b) Concurrent dephasing at $a=0.7$ and a second admissible profile $S=1-\gamma/2$ illustrate asymmetric transverse decay. (c) For two qubits the deviation $\gamma_e+\gamma_+-1$ is plotted against $r=\alpha/\beta$ over the range where magic returns.}
\label{si:figprofiles}
\end{figure}

The two threshold surfaces $Q=r^2$ and $G=r^2$ are exchanged by Eq.~\eqref{si:complement}. One can see the reflection at each damping strength. The reservoir coherence is $\alpha\beta\sqrt G$, and on every cut the two diagonal entries of its partial-transpose block have product $\beta^4GQ$. For $0<\gamma_i<1$ this gives
\[
 \frac{(\alpha\beta\sqrt G)^2}{\beta^4GQ}
 =\left(\frac{c}{P_n}\right)^2=\frac{r^2}{Q}.
\]
The reservoir block has a negative eigenvalue exactly when the register coherence overtakes its fully excited population. This holds across every cut even when the damping strengths differ. With equal damping strengths the thresholds are $\gamma_+=1-r^{2/n}$ and $\gamma_e=r^{2/n}$. The order of the initial magic loss, its return and entanglement death is illustrated in Fig.~\ref{fig:mechanism}b. To obtain the boundaries evaluate the strictly increasing difference $P_0-c$ at $\gamma_e$. Its value is $\beta^2 r[2r-(1-r^{2/n})^{n/2}]$. Thus $\gamma_e=\gamma_-$ at $r_1=(1+2^{2/n})^{-n/2}$ and $\gamma_e=\gamma_+$ at $r_2=2^{-n/2}$. Increasing $r$ through $r_1$ and $r_2$ moves $\gamma_e$ from before both magic crossings, to between them, to after both.

\subsection{Additional dephasing and preparation phase}
First consider a phase-covariant channel that fixes $\ket0$, sends $\ket1\!\bra1$ to $(1-q)\ket0\!\bra0+q\ket1\!\bra1$ with $q=1-\gamma$, and multiplies $\ket0\!\bra1$ by a real $\lambda(\gamma)$. Complete positivity gives $0\le S(\gamma)=\lambda(\gamma)^2/q\le1$ on the open interval. For a real cat the partial-transpose determinant vanishes when $\gamma=r^{2/n}S(\gamma)$, while $|c|=P_n$ holds when $1-\gamma=r^{2/n}S(\gamma)$. Suppose both crossings exist and are unique in $(0,1)$. These two equations give
\begin{equation}
 \gamma_++\gamma_e=1\quad\Longleftrightarrow\quad S(\gamma_e)=S(1-\gamma_e).
 \label{si:profile}
\end{equation}
Indeed the right side places $1-\gamma_e$ at the excited-endpoint root; there $P_n=\beta^2\gamma_e^n\le\alpha^2\le P_0$. A fixed dephasing layer gives a constant profile and preserves this threshold sum when both crossings exist. Concurrent Markovian dephasing gives $S(\gamma)=(1-\gamma)^{2a-1}$, which is strictly monotone for $a>1/2$. Its reflected value can then agree only at $\gamma_e=1/2$. Figure~\ref{si:figprofiles} compares these cases.

The constant profile also distinguishes the threshold sum from the vacuum system--environment symmetry. Take a pure real cat and include a fixed dephasing layer in the channel, with single-qubit coherence multiplier $\lambda(\gamma)=\sqrt{\eta(1-\gamma)}$ for the combined channel. Denote its register output by $\rho_\eta(\gamma)$. The layer multiplies the cat coherence by $\eta^{n/2}$ and leaves the populations unchanged. Using the equal purities for pure amplitude damping in Eq.~\eqref{si:complement} gives
\[
 \operatorname{Tr}\!\big[\rho_\eta(\gamma)^2\big]
 -\operatorname{Tr}\!\big[\rho_\eta(1-\gamma)^2\big]
 =2\alpha^2\beta^2(1-\eta^n)\big[\gamma^n-(1-\gamma)^n\big].
\]
For $0<\eta<1$ the purities differ whenever $\gamma\ne1/2$. In a pure dilation the full environment has the same nonzero spectrum as the register, so it cannot reproduce the reflected register output at those damping strengths.

For a cat prepared with phase $n\phi$, the positive coherence magnitude $|c|=\alpha\beta q^{n/2}$ under pure amplitude damping is multiplied in the stabilizer test by $F=|\cos(n\phi)|+|\sin(n\phi)|$, as follows from Eq.~\eqref{si:xrom}. The magnitude and the entanglement threshold are independent of $\phi$. The excited-endpoint crossing moves to $\gamma_+=1-(rF)^{2/n}$. To have a nonzero stabilizer interval, one needs $rF<1$ and $P_0>F|c|$ at this crossing, namely
\begin{equation}
 [1-(rF)^{2/n}]^n>r^2(F^2-1).
\end{equation}
Equality allows the trajectory to touch the stabilizer boundary at a single instant. If a finite interval exists, its reflected-threshold deficit is $1-\gamma_e-\gamma_+=r^{2/n}(F^{2/n}-1)$, which vanishes exactly at total phases $n\phi\in(\pi/2)\mathbb Z$. Figure~\ref{si:figphase} gives the phase dependence in a two-qubit example.

\begin{figure}[tbp]
\centering\includegraphics[width=.82\linewidth]{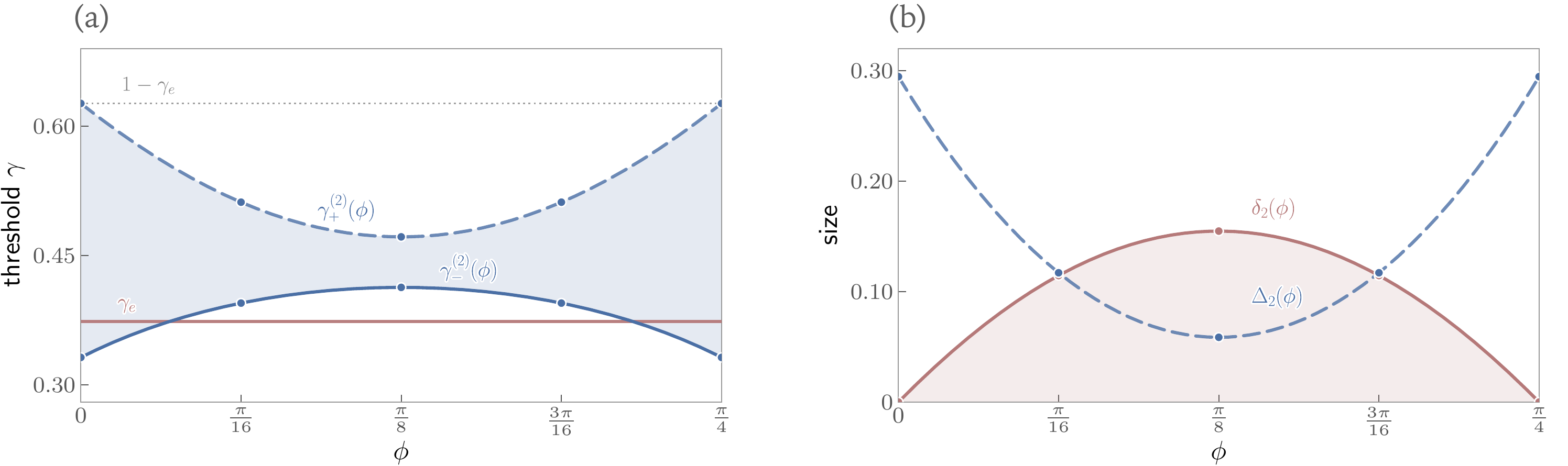}
\caption{\textbf{A relative phase moves the magic thresholds.} The input has $n=2$, $\alpha=0.35$ and total phase $2\phi$. (a) The entanglement threshold stays at $\gamma_e=r$ as $\phi$ changes. The shaded region lies between magic death and its return. (b) The loss interval $\Delta_2=\gamma_+-\gamma_-$ is narrowest and the deficit $\delta_2=1-\gamma_e-\gamma_+$ largest at $\phi=\pi/8$ for the range shown.}
\label{si:figphase}
\end{figure}

\section{Conditioned parity readout and errors in the classical record}
\label{si:conditional}
The joint decoder in the article returns a qubit on every run. Here we quantify what can be certified if only the all-even parity outcome is retained, including the effect of errors in the reported outcomes. In an ideal measurement of the repetition parities $Z_iZ_{i+1}$, the probability of this outcome is $s=P_0+P_n$. Decoding gives a qubit with Bloch vector $(2c,0,P_0-P_n)/s$ when $c\ge0$. All other syndromes are diagonal for the cat trajectory. Thus the probability-weighted excess on the retained runs is
\begin{equation}
 s[\mathcal R(\tilde\rho)-1]
 =[2c+|P_0-P_n|-s]_+=2[c-\min(P_0,P_n)]_+.
 \label{si:yield}
\end{equation}
With pure amplitude damping on the reborn branch, the right side is $2(\alpha\beta v-\beta^2v^2)$ for $v=q^{n/2}$. It reaches $\alpha^2/2$ at $v=r/2$. For the zero-temperature channel with dephasing and $1/2\le a<1$, it becomes $2(\alpha\beta q^{an}-\beta^2q^n)$ and peaks at $q^{(1-a)n}=ar$, with value $2(1-a)\alpha\beta(ar)^{a/(1-a)}$.

Figure~\ref{si:figconditional} shows the normalized conditional qubit for three register sizes. The pure-state arc has a simple limit under pure amplitude damping. Hold the total exposure $x=-n\ln q$ fixed as $n$ increases. Then $\gamma^n=(1-e^{-x/n})^n\to0$, while $q^n=e^{-x}$. The all-even outcome has limiting probability $\alpha^2+\beta^2 e^{-x}$, and the conditional output tends to the pure state
\[
 \frac{\alpha\ket0+\beta e^{-x/2}\ket1}{\sqrt{\alpha^2+\beta^2 e^{-x}}}.
\]
At $x=2\ln[(1+\sqrt2)/r]$ the no-jump amplitude ratio is $\beta q^{n/2}/\alpha=\tan(\pi/8)$. This branch therefore decodes to the $+1$ Hadamard eigenstate $\ket H=\cos(\pi/8)\ket0+\sin(\pi/8)\ket1$ even at finite $n$. The same parity record also occurs when all excitations are lost. Writing $s_H=(4-2\sqrt2)\alpha^2$ for the probability of the no-jump branch, the accepted output obeys
\[
 s\tilde\rho=s_H\ket H\!\bra H+\beta^2\gamma^n\ket0\!\bra0,\qquad
 s=s_H+\beta^2\gamma^n .
\]
Its infidelity is $\beta^2\gamma^n\sin^2(\pi/8)/s$. For a balanced initial cat ($\alpha=\beta$), damping generates magic from a stabilizer input. At this exposure the acceptance probability tends to $2-\sqrt2$ as $n$ increases, while the probability-weighted excess in Eq.~\eqref{si:yield} equals $3\sqrt2-4$ at every $n$.
\begin{figure}[tbp]
\centering\includegraphics[width=.66\linewidth]{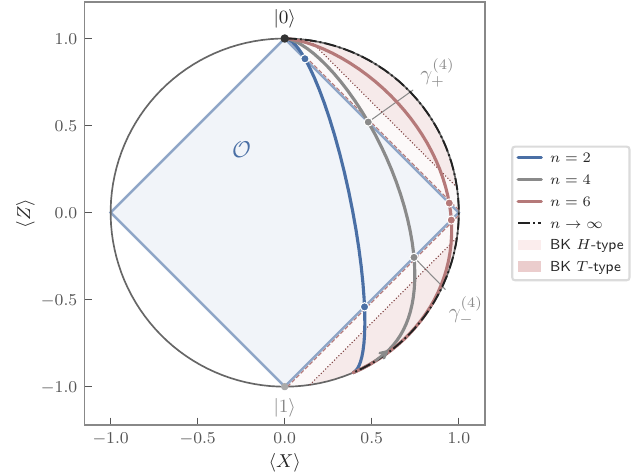}
\caption{\textbf{Decoded qubit conditioned on the all-even parity outcome.} The curves for $n=2,4,6$ and $\alpha=0.20$ give the Bloch coordinates of the retained qubit as damping increases. The diamond $\mathcal O$ is the stabilizer octahedron in the $XZ$ plane. The two crossings are labeled for $n=4$. The dash-dotted arc gives the pure-state limit at fixed total exposure. The pink regions give the sufficient thresholds of the particular Bravyi--Kitaev protocols for $H$-type and $T$-type states~\cite{Bravyi2005Universal}. With unrestricted ideal stabilizer operations, the distillation threshold in this plane is the edge of the diamond~\cite{Reichardt2005Quantum}.}
\label{si:figconditional}
\end{figure}

Suppose each of the $n-1$ reported parity bits flips independently with probability $\epsilon\le1/2$ after ideal measurement; all subsequent Clifford gates act ideally. The correctly reported all-even outcome contributes weight $A=(1-\epsilon)^{n-1}s$. A syndrome with $k\ge1$ negative bits is falsely recorded as all-even with probability $\epsilon^k(1-\epsilon)^{n-1-k}\le\epsilon(1-\epsilon)^{n-2}$. Hence the total incorrectly accepted weight is at most $J_{\max}=\epsilon(1-\epsilon)^{n-2}(1-s)$. Each incorrectly accepted qubit is diagonal after the computational-basis decoder. For actual incorrectly accepted weight $J$ the unnormalized $x$ component of the mixture is $2Ac/s$ and its $z$ component has magnitude at least $A|P_0-P_n|/s-J$. It follows that
\begin{equation}
 (A+J)[\mathcal R(\tilde\rho_\epsilon)-1]
 \ge [A(\mathcal R(\tilde\rho)-1)-2J]_+.
\end{equation}
On the reborn branch $c>P_n$, so using $J\le J_{\max}$ and Eq.~\eqref{si:yield} gives the sufficient certificate
\begin{equation}
 \mathcal R(\tilde\rho_\epsilon)-1\ge
 \frac{2[(1-\epsilon)(c-P_n)-\epsilon(1-s)]_+}{(1-\epsilon)s+\epsilon(1-s)}.
 \label{si:readout}
\end{equation}
The bound is positive when $\epsilon<\epsilon_c=(c-P_n)/(c-P_n+1-s)$. At the ideal peak $q^{n/2}=r/2$, $c-P_n=\alpha^2/4$, $P_n=\alpha^2/4$ and, for fixed $r$ as $n\to\infty$, $P_0\to\alpha^2$. Hence $\epsilon_c\to r^2/4$, about $4.8\%$ for $\alpha=0.4$. This is a sufficient tolerance for certifying the conditional output. The correctly accepted probability still includes $(1-\epsilon)^{n-1}$ and decreases with $n$ at fixed nonzero $\epsilon$.

\paragraph{A sufficient error budget for the deterministic decoder.}
The bit-flip model above does not describe faulty gates or measurement backaction. For the deterministic joint decoder of the article, a sufficient condition follows from a trace-distance bound. On the reborn branch its ideal output has $\langle X+Z\rangle=\mathcal R>1$, whereas every single-qubit stabilizer mixture has $\langle X+Z\rangle\le1$. Since $X+Z$ has operator norm $\sqrt2$, this expectation differs between two states by at most $2\sqrt2$ times their trace distance. An implemented output therefore has magic whenever its trace distance from the ideal output is below $(\mathcal R-1)/(2\sqrt2)$. For a fixed input, this distance is at most the sum, over the elements of the circuit, of half the diamond-norm distance between each implemented operation and its ideal counterpart, with every measurement described as an instrument that includes its outcome register. At the pure-damping peak for $\alpha=0.4$, $\mathcal R-1=0.080$ and a trace distance below $0.028$ suffices.

\section{Extraction by separated laboratories}
\label{si:local}
\subsection{One copy shared between sites}
\label{S:access:protocol}

We consider pure amplitude damping of the real cat with
$0<\alpha/\beta<1$ and $n\ge2$. For $q=1-\gamma$ write
$D_\gamma=\gamma\ket0\!\bra0+q\ket1\!\bra1$ and
$c=\alpha\beta q^{n/2}$. The state shared by the laboratories is
\begin{equation}
 \rho=\alpha^2\ket{0^n}\!\bra{0^n}+\beta^2D_\gamma^{\otimes n}
       +c\bigl(\ket{0^n}\!\bra{1^n}+\ket{1^n}\!\bra{0^n}\bigr).
 \label{S:eq:access:rho}
\end{equation}
The laboratories have one copy and no shared ancillary entanglement.
They can prepare local stabilizer ancillas, apply local Clifford gates,
measure local Pauli operators and exchange classical messages in either
direction. They may discard systems and postselect outcomes. We first
bound the output on each branch of any finite-round protocol. This
also bounds the state obtained by retaining several outcomes or
forgetting the record.

\paragraph{Conditioning the protocol.}
Fix the entire classical record, including any random choices and a
computational-basis outcome for each discarded qubit. The latter
refinement is valid since summing the outcomes implements a partial
trace. All feedforward choices are now fixed. Operators in separate
laboratories commute, so the branch operator factors into one Kraus
operator per laboratory. At the receiving laboratory it maps the local
input to one qubit. At every other laboratory it is a bra on the
local input.

For any of these local operators $K$, form its unnormalized Choi
vector by applying $K$ to one half of a Bell pair for each input
qubit. Initially the Bell pairs and local ancillas are pure
stabilizer states. Each Clifford maps such a state to another, and
projecting onto a fixed Pauli measurement outcome gives either zero
or a scalar multiple of a pure stabilizer state. Contracting all
measured qubits with the specified final outcomes therefore gives a
stabilizer Choi vector up to a scalar on every nonzero branch. In
particular, a laboratory retaining no output applies a bra
proportional to $\bra s$ for some pure stabilizer state $s$ (complex
conjugation in the Choi convention preserves the stabilizer set).
For a fixed branch, one can first contract all these bras with the
input and then act with the receiver's Kraus operator, even when the
original messages alternated between laboratories.

\paragraph{One qubit at each site.}
Each assisting bra projects onto an eigenstate of $X$, $Y$, or $Z$.
A $Z$ projection removes the coherence in
Eq.~\eqref{S:eq:access:rho}, leaving the receiver with a
computational-basis mixture. Otherwise, each of the $n-1$ bras has
the form $(\bra0+e^{-i\phi_j}\bra1)/\sqrt2$ with
$\phi_j\in\{0,\pi/2,\pi,3\pi/2\}$. The squared overlap with either
basis value is $1/2$. Contracting Eq.~\eqref{S:eq:access:rho} with
these bras gives a factor $2^{-(n-1)}$ on every diagonal contribution
and a known phase $e^{i\sum_j\phi_j}$ on the endpoint coherence.
After normalization and a single-qubit Clifford phase correction the
state entering the receiver's operation is therefore
\begin{equation}
 \begin{pmatrix}\alpha^2+\beta^2\gamma & c\\ c&\beta^2q\end{pmatrix}.
 \label{S:eq:access:siteconditional}
\end{equation}
Every refined receiver operation has one input and one output. Its
pure two-qubit stabilizer Choi vector is either a product of two
Pauli eigenstates or is locally Clifford-equivalent to a Bell pair.
Indeed, a local nonidentity stabilizer of the Choi vector fixes one
marginal to be pure and forces a product. In its absence the two
independent stabilizers have anticommuting local Pauli components on
both qubits, which a local Clifford maps to the Bell generators
$XX$ and $ZZ$, with Pauli corrections for their signs. A product
Choi vector prepares a fixed stabilizer output. A Bell Choi vector
up to local Cliffords implements a Clifford up to a scalar, since
transposition also takes a Clifford to a Clifford up to phase.
Consequently the receiver cannot increase the conditional robustness
on any such branch.

For a real single-qubit state with coherence $c\ge0$ and populations
$p_0,p_1$, the robustness excess is $2[c-\min(p_0,p_1)]_+$. Here
$c\le\beta^2q$ since $(\alpha/\beta)q^{n/2-1}\le1$. The per-branch
upper bound, also valid by convexity for a selected set of outcomes, is
\begin{equation}
 \mathcal R_{\rm site}-1=2[c-(\alpha^2+\beta^2\gamma)]_+ .
 \label{S:eq:access:site}
\end{equation}
Measuring $X$ at every assisting site and applying the phase correction
produces Eq.~\eqref{S:eq:access:siteconditional} for every outcome.
Thus Eq.~\eqref{S:eq:access:site} is also attained deterministically.
On the open reborn branch $q^{n/2}<\alpha/\beta$ and $c<\alpha^2$,
so every one-qubit output at a separated site is a stabilizer mixture.

\paragraph{One qubit per site at finite temperature.}
The same branch calculation applies to the thermal cat trajectory. Its only
coherence is the real endpoint element $c=\alpha\beta q^{an}$. Equatorial
projections have equal squared overlap with every computational-basis
string, so the two diagonal entries in the normalized conditional state are
the receiver's marginal populations $p_0$ and $p_1$. Its corrected coherence
is $c$. An assisting $Z$ projection still removes this element. The exact
single-copy optimum is therefore $2[c-\min(p_0,p_1)]_+$ in robustness
excess, including for a postselected branch. Here
\begin{equation}
 p_1=\alpha^2 U+\beta^2 B=\vartheta+(\beta^2-\vartheta)q,
 \qquad U=\vartheta(1-q),\quad B=U+q,\quad p_0=1-p_1 .
 \label{S:eq:access:thermal}
\end{equation}
For $n\ge2$, $a\ge1/2$ and $0<\alpha<\beta$, one has
$c<\beta^2 q\le\beta^2 B\le p_1$ at every finite time. At global
magic death, the all-zero population satisfies $P_0\ge c$, and $p_0\ge P_0$
because it sums over all strings with the receiver in state $0$.
Since $\beta^2>1/2\ge\vartheta$, Eq.~\eqref{S:eq:access:thermal} shows that
$p_0$ grows as $c$ decays. Every subsequent single-site output therefore
has zero excess. For the two-qubit example $a=0.55$, $\vartheta=0.25$,
$r=0.95$ and $q=0.4$, one obtains $p_1=0.36025$ and $c=0.18225$.
The global excited population $P_{11}=0.16967$ gives a jointly extractable
excess of $0.02515$, while every sitewise output is a stabilizer mixture.
The state is separable here, since $P_{01}=P_{10}=0.19057>c$.

For an endpoint coherence with arbitrary phase, equatorial Pauli outcomes
multiply $c$ by powers of $i$, preserving
$\nu(c)=|\operatorname{Re}c|+|\operatorname{Im}c|$. A one-qubit
Clifford leaves the Bloch-vector one-norm invariant. The same branch
argument therefore gives $2[\nu(c)-\min(p_0,p_1)]_+$ for the optimum at
one site. For the complex cat $(3 e^{i\pi/4}\ket{00}+5\ket{11})/\sqrt{34}$
at $a=0.55$, $\vartheta=1/4$ and $\gamma=0.55$, both marginal
populations exceed $\nu(c)$, while the joint robustness excess is
$8.39\times 10^{-4}$.

\paragraph{Two laboratories.}
Let the receiver hold $k$ qubits and the assisting laboratory hold
$h=n-k\ge1$. Consider the open reborn branch at a finite time,
so $0<q<1$ and $c<\alpha^2$. If the assisting stabilizer bra has zero
overlap with either $0^h$ or $1^h$, the receiver is left with a
diagonal state. Otherwise write $L$ for its computational-basis
support. Stabilizer amplitudes have constant magnitude on $L$ and
relative phases in $\{1,i,-1,-i\}$~\cite{Dehaene2003Clifford}.
Contract the assistant's bra and divide the resulting unnormalized
receiver state by the common squared endpoint overlap. A Clifford
phase correction makes its sole coherence $c$ real. Its two endpoint
populations are
\begin{align}
 d_0&=\alpha^2+\beta^2\gamma^k W_L, &
 d_1&=\beta^2 q^k W_L, &
 W_L&=\sum_{y\in L}\gamma^{h-|y|}q^{|y|} .
 \label{S:eq:access:support}
\end{align}
We have $d_0\ge\alpha^2>c$ and
$W_L\ge\gamma^h+q^h$ because $0^h,1^h\in L$.
A positive matrix with just this real coherence is a subnormalized
stabilizer mixture whenever $c\le\min(d_0,d_1)$: allocate weight
$2c$ to $(\ket{0^k}+\ket{1^k})/\sqrt2$ and distribute the rest over
computational-basis projectors. Conversely, every stabilizer mixture
satisfies this inequality because each pure stabilizer state's nonzero
amplitudes have equal magnitude. Thus if
$c\le\beta^2q^k(\gamma^h+q^h)$, the receiver starts its branch with
a stabilizer mixture. Its local stabilizer instrument preserves this
property on each subsequent outcome as well as after combining them.

The bound is tight: the assistant can project onto
$(\ket{0^h}+\ket{1^h})/\sqrt2$ by measuring its commuting Pauli
stabilizers. This choice gives $L=\{0^h,1^h\}$. The receiver projects
onto the span of its own two endpoints by measuring $Z_iZ_{i+1}$
within the laboratory, and decodes the all-even outcome to one qubit
by Clifford gates (for $k=1$ this step is the identity). This branch
has nonzero probability, and its normalized output lies outside the
octahedron whenever $c>d_1$. With $P_n=\beta^2q^n$, the necessary
and sufficient condition is therefore
\begin{equation}
 c>P_n\bigl[1+(\gamma/q)^h\bigr].
 \label{S:eq:access:two}
\end{equation}
Equality gives a stabilizer output on every branch. When $k=n$ a
single laboratory holds the entire register, and can decode the
reborn magic at every point of the open branch.

Finally set $r=\alpha/\beta$. Full separability on the reborn
branch gives $q^{n/2}<r\le\gamma^{n/2}$, hence $\gamma>q$.
Equation~\eqref{S:eq:access:two} requires
$r>q^{n/2}[1+(\gamma/q)^h]$. For $h\ge n/2$ the right-hand side is
at least $q^{n/2}+\gamma^{n/2}>r$. If $h<n/2$, it tends to zero as
$q\to0$, so every strict majority succeeds at sufficiently late finite
times. The limiting ground state itself has zero magic.

\subsection{Multiple copies shared between sites}
Each site holds one qubit from each of $m$ independent copies, with no
shared ancillary entanglement. First consider a real cat pair on the
reborn branch under identical relaxation. Its populations satisfy
$p=P_{01}=P_{10}$ and $P_{11}<c<P_{00}$. The partial transpose has
eigenvalues $P_{00},P_{11},p+c,p-c$, so the pair is entangled exactly
when $c>p$.

Each site measures $Z_iZ_{i+1}$ on successive copies and keeps the
all-even outcome. Local controlled-NOT gates decode $\ket{0^m}$ and
$\ket{1^m}$ to one logical qubit. The resulting pair has unnormalized
entries $\rho_{uv}^{\,m}$, since the filter retains the same two-bit
string in each copy. Measuring $X$ at one site and correcting the sign
at the other gives a qubit with conditional robustness excess
\begin{equation}
 \frac{2[c^m-p^m-P_{11}^m]_+}{P_{00}^m+2p^m+P_{11}^m}.
 \label{S:eq:multiple_excess}
\end{equation}
The denominator is the probability that both sites accept. If $c>p$,
then $p/c$ and $P_{11}/c$ are both less than one, so a finite $m$
makes the numerator positive. For pure amplitude damping the criterion
reduces to $r^m>\gamma^m+q^m$. At the star in Fig.~\ref{fig:channels},
with $\gamma=0.50$, five copies give an excess of $0.0103$ with
acceptance probability $0.0167$. This parity protocol gives zero excess
with four or fewer copies at that point.

When $c\le p$, the pair is separable and the partial-transpose argument
in Methods (Stabilizer readout) rules out extraction by any local
stabilizer protocol, for any finite number of copies. Thus entanglement
is both necessary and sufficient for extraction on the reborn branch
of this pair.

\end{document}